%% file: acl.tex
\documentclass[11pt]{article}

\usepackage{acl}

\usepackage{times}
\usepackage{latexsym}

\usepackage[T1]{fontenc}

\usepackage[utf8]{inputenc}

\usepackage{microtype}

\usepackage{inconsolata}

\usepackage{booktabs}
\usepackage{multirow}
\usepackage{graphicx}
\usepackage[table]{xcolor}
\usepackage{subcaption}
\usepackage{amsmath, amssymb, amsthm}
\usepackage{xcolor}
\usepackage[ruled,linesnumbered]{algorithm2e}
\definecolor{sftblue}{rgb}{0.0706, 0.5882, 0.8588}
\definecolor{mygreen}{rgb}{0.0859375,0.546875,0.0859375}

\newtheorem{theorem}{Theorem}

\newcommand{\TheName}{TSPORec}
\title{\TheName{}: Token Selection via Preference Optimization for LLM-Based Sequential Recommendation}

\author{Wenqiao Zhu \thanks{\,\,\textbf{Correspondence:} \href{mailto:zhuwnq@outlook.com}{zhuwnq@outlook.com}} \\
  \And
  Chao Xu \\
  \And
  Haipang Wu \\
  \And
  Ji Liu}

\begin{document}
\maketitle
\begin{abstract}
Large Language Models (LLMs) have emerged as powerful tools for improving recommendation systems.
The effectiveness of LLMs arises from their ability to harness rich textual information and their capacity to model heterogeneous user preferences based on users’ interaction history.
However, due to the large-scale and deep architectures, LLM-based sequential recommendation approaches generally incur high inference costs, resulting in a low return on investment. 
To mitigate this cost, many existing approaches resort to using only the first few tokens of item descriptions, which inadvertently discards valuable information contained in the full text, thereby leading to suboptimal recommendation performance.  
To address this limitation, we propose a novel Token Selection approach for Preference Optimization in LLM-based sequential Recommendation, i.e., \TheName{}, which accurately pinpoints informative tokens throughout the entire textual content to improve recommendation performance. 
Specifically, we design a three-stage pipeline to select informative tokens and introduce a novel proxy reward to facilitate the implementation. 
\TheName{} not only enhances recommendation performance but also improves computational efficiency.
Extensive experiments across two models and 
datasets demonstrate the superb performance (up to 31.25\%) and efficiency (up to 63.4\%) of our approach compared with six baseline approaches. Code is avaliable at https://github.com/WNQzhu/TSPORec.git.
\end{abstract}

\section{Introduction}

Recommendation systems have become indispensable for mitigating information overload across a wide range of commercial platforms \cite{hstu,10.1145/3746252.3761507,zhu2025csdm,10.1145/3511808.3557704}. A key factor in delivering personalized item recommendations lies in accurately modeling the heterogeneous preferences of users. Recently, Large Language Models (LLMs) have demonstrated remarkable effectiveness in capturing such preferences \cite{10.1145/3637528.3671931, 10.1145/3626772.3657690, 6c48a0f1d7d84077a16ce55105dc8ddc}, owing to their vast parameter scales and strong contextual understanding. In particular, Sequential Recommendation (SR), which models diverse interests of users based on their recent interaction histories, has proven highly effective in enhancing recommendation performance when integrated with LLMs \cite{HLLM, liu2024large, HLLM-Creator}.

The goal of SR is to predict the next items a user is likely to prefer by modeling sequential patterns in their historical interactions. A key challenge in this field lies in simultaneously capturing item-level sequential dependencies and modeling the complex and heterogeneous preferences of users derived from these interaction sequences.

Owing to the pivotal role in advancing recommendation systems, SR has attracted sustained attention from both academia and industry \cite{rendle10mc, cheng13poi, HLLM, liu2024large}. Existing approaches \cite{kang2018selfattentive, 10.1145/3336191.3371786, 10.1145/3383313.3412258} predominantly leveraged item IDentification (ID) features, overlooking the abundant textual content typically present in real-world settings. Furthermore, their reliance on shallow network architectures limit their capacity to model the complex and heterogeneous nature of user preferences.

Recently, LLMs have emerged as effective solutions for overcoming the limitations of previous approaches, demonstrating strong performance in SR tasks \cite{HLLM, liu2024large, HLLM-Creator}. The core idea behind these approaches is to leverage the rich textual features and powerful representational capabilities of LLMs to extract robust representations of items and users. For example, HLLM \cite{HLLM} employs separate item and user LLMs, which are based on LLMs such as TinyLlama \cite{Zhang2024TinyLlamaAO} or Baichuan \cite{baichuan2}, to learn comprehensive item and user representations.

However, existing LLM-based SR approaches are highly resource-intensive. For instance, given an interaction history of length $S$ with each item represented by $n$ tokens, a single inference step requires processing $S \times n$ tokens, which consists of both computation-intensive and memory-intensive operations. To mitigate this cost, existing approaches \cite{HLLM, liu2024large} typically exploit only the first few tokens of each item, failing to fully exploit their rich textual features. This limitation raises an important research question: \textit{Can we identify and select the most informative tokens from textual features to maximize information utilization?} Addressing this question could benefit SR systems in two key ways. First, it enables effective usage of critical information within the text. Second, it reduces computational overhead through compressing the overall token sequence length.

In this paper, we propose a novel token selection approach for preference optimization in LLM-based Sequential Recommendation, i.e., \textbf{\TheName{}}. \TheName{} is composed of three stages to select the most informative tokens from textual features. First, we pre-train an SR model based on an LLM. Next, we freeze the LLM backbone and attach a policy head, which is then trained using a newly designed proxy reward function.  Using the trained policy, we extract informative tokens from item texts and subsequently retrain the LLM on the refined input sequences. We conduct extensive experiments across six recommendation approaches and two publicly available benchmark datasets to demonstrate the effectiveness and generality of our approach. The key contributions of this paper are summarized as follows:
\begin{itemize}
    \item We propose a new token selection pipeline that automatically identifies informative tokens from textual features. This pipeline shows strong generalization across LLMs and datasets.

    \item We design a novel proxy reward method to guide policy learning in an effective and targeted manner. This method can identify meaningful token importance, leading to an effective policy.

    \item Extensive experimental results demonstrate the superb performance (up to 29.29\% in terms of recall and 31.25\% in terms of NDCG) and efficiency (up to 63.4\% inference overhead reduction) of \TheName{} compared with six State-of-the-Art baseline approaches.
\end{itemize}

\begin{figure*}[t]
\centering
\includegraphics[width=\linewidth]{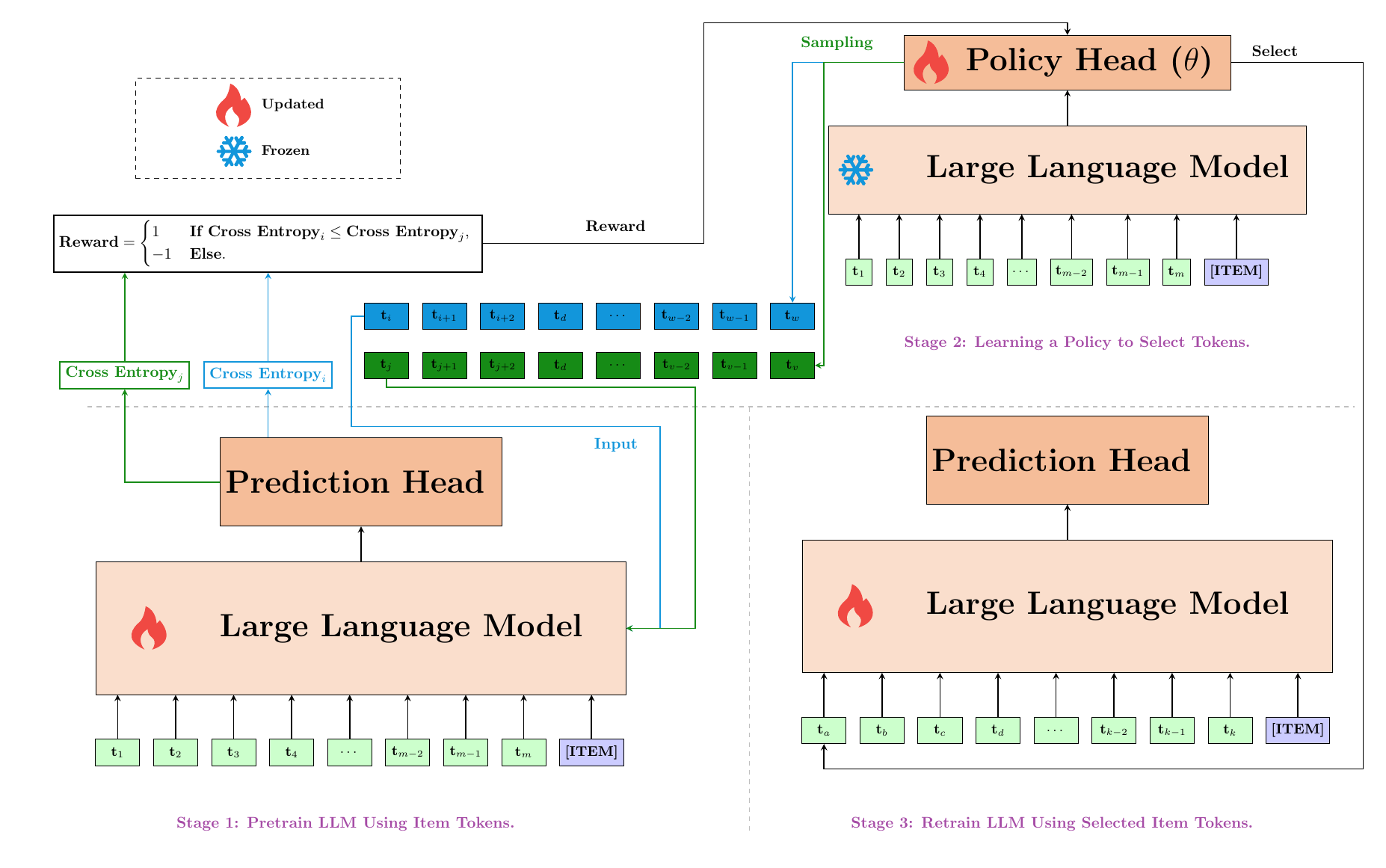}
\vspace{-10mm}
\caption{Overview of \TheName{}'s Three-Stage Training Pipeline: (1)~LLM Pretraining, (2)~Token Selection Policy Learning, (3)~Token Selection and Retraining.}
\vspace{-6mm}
\label{fig:framework}
\end{figure*}
\section{Preliminaries}
\label{sec:prelim}
Modern recommendation systems typically represent users and items as learned embedding vectors. Let $\mathbf{e}_u$ and $\mathbf{e}_i$ denote the embeddings of a user and an item, respectively. In SR, the user embedding is derived from their interaction history:
$\mathbf{e}_u = g(\{\mathbf{e}_i\}),$
where $g(\cdot)$ is commonly a shallow neural network or a LLM and $\mathbf{e}_i$ represents a positive item. A positive item $\mathbf{e}_i$ for User $u$ is an item that User $u$ has interacted with. In this work, we denote the LLM exploited to aggregate item embeddings into a user representation as the \textbf{user LLM}, and the LLM employed to generate item embeddings via token selection as the \textbf{item LLM}.

A widely adopted training objective in recommendation systems is the InfoNCE loss \cite{infonce}, as defined in Formula~(\ref{eq:infonce}):
\begin{equation}
  \label{eq:infonce}
    \scalebox{0.85}{$\displaystyle
      \mathcal{L}_{\text{InfoNCE}}(\mathbf{e}_i) = -\log \frac{\exp(\langle \mathbf{e}_u, \mathbf{e}_i \rangle)}{\exp(\langle \mathbf{e}_u, \mathbf{e}_i \rangle) + \sum_{j=1}^N \exp(\langle \mathbf{e}_u, \mathbf{e}_i^j \rangle)},
      $}
\end{equation}
where $\langle \cdot, \cdot \rangle$ denotes a similarity measure, i.e., dot product, and {\small $\{\mathbf{e}_i^1, \dots, \mathbf{e}_i^N\}$} are embeddings of $N$ negative items. Negative items are typically the items that the user has not interacted with. We denote {\small $\mathbf{E}(\mathbf{e}_i) = \big[ \langle \mathbf{e}_u, \mathbf{e}_i \rangle, \langle \mathbf{e}_u, \mathbf{e}_i^1 \rangle, \ldots, \langle \mathbf{e}_u, \mathbf{e}_i^N \rangle \big]$} the vector of similarity scores between the user embedding and the positive item embedding along with $N$ negative items. We take $\mathit{label} = [1, 0, \ldots, 0]$ as the corresponding one-hot label vector. The InfoNCE loss can then be equivalently expressed as a cross-entropy objective:
\begin{equation}
\label{eq:ce}
\mathcal{L}_{\text{InfoNCE}}(\mathbf{e}_i) = \mathcal{L}_{\mathrm{CE}}\big( \mathrm{softmax}(\mathbf{E}(\mathbf{e}_i)), \mathit{label} \big),
\end{equation}
where $\mathcal{L}_{\mathrm{CE}}$ denotes the Cross-Entropy loss.

\section{Method}
\label{sec:method}

In this section, we first formulate the problem to address in this paper. Then, we detail our approach, i.e., \TheName{}, including the new token selection pipeline and the novel proxy reward method.

\subsection{Problem Formulation}
\label{sec:problem}

We study the problem of token selection in LLM-based SR. Given an item represented by a token sequence $T = \{t_1, t_2, \dots, t_n\}$, conventional LLM-based approaches typically derive an item embedding using only the first $k$ tokens:
\[
\mathbf{e}_i = \text{LLM}(t_1, t_2, \dots, t_k),
\]
where $k \leq n$. Instead, we aim to identify an informative subset of $k$ tokens, denoted as $T_k \subseteq T$, and feed them to the LLM to compute a refined item embedding:
$
\mathbf{e}'_i = \text{LLM}(T_k).
$
Our goal is to enhance downstream recommendation performance by leveraging $\mathbf{e}'_i$, achieving superior results compared to those obtained using the original prefix-based embedding $\mathbf{e}_i$. The selection of $T_k$ is guided by a learnable policy, which prioritizes tokens that contribute most to user preference modeling. The problem addressed in this paper can be formulated as maximizing the expected reward:
\begin{equation}
\label{eq:problem}
\max_{T_k} \mathcal{R}(\theta).
\end{equation}
We defer the definition of the reward $\mathcal{R}(\theta)$ to the next section, where it is given in 
Formula~(\ref{eq:expected-reward}).
\subsection{Informative Tokens Selection}
\label{sec:pipeline}

In this section, we detail our approach, i.e., \TheName{}. To maximize flexibility, we perform informative token identification at the \textit{chunk level}. A chunk is a subset of tokens consisting of several consecutive tokens in a token sequence. Specifically, we select $\lfloor k/c \rfloor$  chunks from the token sequence $T$, where each chunk consists of $c$ consecutive tokens:
\begin{equation*}
  \scalebox{0.85}{$\displaystyle
    \left[ (t_1^1, t_1^2, \dots, t_1^c),
      \ldots,\, (t_{\lfloor k/c \rfloor}^1, t_{\lfloor k/c \rfloor}^2, \dots, t_{\lfloor k/c \rfloor}^c) \right],
    $}
\end{equation*}
where $c$ denotes the predefined chunk size, and $\lfloor \cdot \rfloor$ is the floor function, which returns the greatest integer less than or equal to its argument.  Token-level selection is a special case of this formulation, corresponding to $c = 1$.  This chunk-wise strategy allows for controllable granularity in capturing local semantic structures for downstream recommendation tasks. As shown in Figure~\ref{fig:framework}, the overall pipeline of \TheName{} consists of three stages: \textit{Pretraining}, \textit{Token Selection Policy Learning}, and \textit{Token Selection and Retraining}.

\paragraph{Pretraining}

In this stage, we pretrain both the user LLM (denoted as $g$) and the item LLM (denoted as $f$) employing the InfoNCE loss defined in Eq.~(\ref{eq:infonce}) to obtain a foundation model that can be leveraged for reward computation.

For item LLM pretraining, we append a special \textit{[ITEM]} token to the end of the textual token sequence of each item. The item embedding is then extracted as the final-layer hidden state corresponding to this \textit{[ITEM]} token. Specifically, given an input sequence $ T = \{t_1, t_2, \dots, t_k, \textit{[ITEM]}\} $, the item LLM produces hidden states:
\begin{equation}
[\mathbf{h}_1, \mathbf{h}_2, \dots, \mathbf{h}_{k+1}].
\end{equation}
We utilize $ \mathbf{h}_{k+1} $, i.e., the representation at the \textit{[ITEM]} position, as the item embedding.

\paragraph{Token Selection Policy Learning}

In this stage, we introduce our novel token selection policy. Given the full set of textual features for all items, exhaustively evaluating every possible token subset is computationally infeasible due to its exponential time complexity. To mitigate this burden, we propose a differentiable proxy reward that approximates the utility of a token subset, enabling efficient optimization of the selection policy.

We begin by quantifying the informativeness of each token $t_i$ through its relevance to the item embedding. 
Specifically, we compute a scalar importance score, referred to as the $\mathrm{info}(\cdot)$, using a query-key attention mechanism:
\begin{equation}
  \mathrm{info}(t_i) = \frac{\langle W_Q \mathbf{h}_i, W_K \mathbf{h}_{k+1} \rangle}{\sqrt{d}},
  \label{eq:logits}
\end{equation}
where $\theta = \{W_Q, W_K\}$ denotes a set of learnable parameter matrices in $\mathbb{R}^{d \times d}$, and $d$ is the dimensionality of the hidden embeddings. 
The resulting score reflects the alignment between the token representation $\mathbf{h}_i$ and the item-level representation $\mathbf{h}_{k+1}$.
We then define the probability to select Token $t_i$ as
\begin{equation}
  p_i = \frac{
    \exp{\left(\mathrm{info}(t_i)\right)}
  }{
   \sum_{j=1}^{k} \exp{\left(\mathrm{info}(t_j)\right)}
    }.
\end{equation}

Given a sequence of $M$ items that a user has interacted with, denoted as $\mathcal{I}_M = \{T_1, T_2, \dots, T_M\}$, where $T_j$ represents the token sequence of the $j$-th item. We model the joint probability of the interaction sequence as
\begin{equation}  
  P(\mathcal{I}_M) = \prod_{j=1}^M \prod_{\ell=1}^k p_{j\ell},
 \end{equation}
where $p_{j\ell}$ is the probability of the $\ell$-th token in $T_j$.

Based on these definitions, we outline the procedure for training the policy parameters \(\theta\). For each token sequence \(T_j\) in the interaction history \(\mathcal{I}_M\) of a user, we perform two independent sampling passes to extract distinct subsets of token chunks, yielding two derived sequences denoted \(T'_j\) and \(T''_j\). This process results in two new token sequence sets:
\begin{align}
  \small
  & \mathcal{I}'_M = \{T'_1, T'_2, \dots, T'_M\} \\
  & \mathcal{I}''_M = \{T''_1, T''_2, \dots, T''_M\}.
\end{align}
Using these sampled sequences, we construct two corresponding user embeddings via the pretrained, frozen LLM:
\begin{align}
  \small
  & \mathbf{e}'_u = g^{\mathrm{frozen}}\big(f^{\mathrm{frozen}}(\mathcal{I}'_M)\big) \\
  & \mathbf{e}''_u = g^{\mathrm{frozen}}\big(f^{\mathrm{frozen}}(\mathcal{I}''_M)\big).
\end{align}
We construct the full item embedding based on the full token sequence as
\begin{equation}
\mathbf{e}_i = f^{\mathrm{frozen}}(\mathcal{I}_M).
\end{equation}
We compute two cross-entropy values by contrasting the user embeddings $\mathbf{e}'_u$ and $\mathbf{e}''_u$ with the full item embedding $\mathbf{e}_i$, exploiting a set of $N$ negative item embeddings $\{\mathbf{e}_i^j\}_{j=1}^N$. This yields two cross-entropy values, $\mathcal{L}_{\mathrm{ce}}^1$ and $\mathcal{L}_{\mathrm{ce}}^2$:
\begin{align}
  \small
  & \scalebox{0.85}{$\displaystyle
      \mathcal{L}_{\mathrm{ce}}^1 = -\log \frac{\exp(\langle \mathbf{e}'_u, \mathbf{e}_i \rangle)}{\exp(\langle \mathbf{e}'_u, \mathbf{e}_i \rangle) + \sum_{j=1}^N \exp(\langle \mathbf{e}'_u, \mathbf{e}_i^j \rangle)}
      $} \label{eq:ce-first} \\
  & \scalebox{0.85}{$\displaystyle
      \mathcal{L}_{\mathrm{ce}}^2 = -\log \frac{\exp(\langle \mathbf{e}''_u, \mathbf{e}_i \rangle)}{\exp(\langle \mathbf{e}''_u, \mathbf{e}_i \rangle) + \sum_{j=1}^N \exp(\langle \mathbf{e}''_u, \mathbf{e}_i^j \rangle)}
      $}. \label{eq:ce-second}
\end{align}

The reward signal $r$ is then defined as:
\begin{equation}
  r = 
  \begin{cases}
    1 & \text{if } \mathcal{L}_{\mathrm{ce}}^1 \leq \mathcal{L}_{\mathrm{ce}}^2, \\
    -1 & \text{otherwise}.
  \end{cases}
  \label{eq:reward-def}
\end{equation}

Our objective is to optimize the policy parameters $\theta$ by maximizing the expected reward:
\begin{equation}
  \scalebox{0.85}{$\displaystyle
    \mathcal{R}(\theta) = \mathbb{E}\!\left[ r \cdot \Big( \log P(\mathcal{I}'_M \mid \theta) - \log P(\mathcal{I}''_M \mid \theta) \Big) \right]
    $}
  \label{eq:expected-reward}
\end{equation}

Given the above description of \TheName{}, we establish the following theorem:

\begin{theorem}
\label{prop:core}
Under the framework of \TheName{}, the following properties hold:
\begin{itemize}
  \item Smaller values of the cross-entropy losses $\mathcal{L}_{\mathrm{ce}}^1$ and $\mathcal{L}_{\mathrm{ce}}^2$ in Eq.~\eqref{eq:ce-first} and Eq.~\eqref{eq:ce-second} yield a tighter approximation to the ground-truth preference distribution in terms of KL divergence.
  
  \item Token chunks shared between the two sampled sequences $\mathcal{I}'_M$ and $\mathcal{I}''_M$ do not contribute to the gradient updates of the policy parameters $\theta$.
  
  \item The objective in Eq.~\eqref{eq:expected-reward} increases the likelihood of selecting informative token chunks while suppressing less informative ones.
\end{itemize}
\end{theorem}

\begin{proof}
We defer the detailed proof to Appendix~\ref{app:proof}.
\end{proof}

\begin{table*}[t]
\small
\begin{tabular}{cc cccc cccc r}
\toprule
Dataset       & Method &  R@5 & R@10 & R@50 & N@5 & N@10 & N@50 &  Impr. (avg)\\
\midrule
\multirow{8}{*}{Amazon Books} & SASRec  & 3.38 & 5.09 & 11.59 & 
2.24 & 2.79 & 4.20  &  +0.0 \%\\
& HSTU        & 2.88 & 4.51 & 11.06 & 1.89   & 2.41 & 3.83 & -11.46\%\\
& LLMinit& 3.29&5.05&11.74&2.18&2.74&4.19& -1.14\%\\
& HLLM(random) &3.29&5.17&12.96&2.17&2.77&4.45& +2.14 \%\\
& HLLM(topk logits) &3.39&5.34&13.35&2.24&2.86&4.59&  +5.36\%\\
& HLLM(first-$k$)   & \underline{4.16} & \underline{6.29}& \underline{14.49}& 
\underline{2.80}&\underline{3.48} & \underline{5.25}&  +24.40\%\\
& \TheName{} (Ours) &  \textbf{4.37} & \textbf{6.55} & \textbf{14.98} &
\textbf{2.94} & \textbf{3.64} & \textbf{5.47} &  +\textbf{29.86}\%\\
& Best Impr. &\cellcolor{lightgray}29.29\%&\cellcolor{lightgray}28.68\%&\cellcolor{lightgray}29.25\%& \cellcolor{lightgray}31.25\%&\cellcolor{lightgray}30.47\%&\cellcolor{lightgray}30.24\%&  \cellcolor{lightgray}N/A\\
\midrule
\multirow{8}{*}{Pixel} & SASRec     & 2.41 & 3.77 & 9.63 &
             1.59 & 2.03 & 3.29 &  +0.0\%\\
& HSTU   & 2.14 & 3.41 & 8.85 &  
 1.40 & 1.81 & 2.97 &  -10.22\%\\
& LLMinit  &2.65&4.10&10.53&1.73&2.19&3.58& +8.92\% \\
& HLLM(random)  &\underline{2.82}&4.39&11.07&\underline{1.84}&2.34&3.78& +15.71\% \\
& HLLM(topk logits) &2.79&4.36&11.02&1.82&2.33&3.76& +14.89\%\\
& HLLM(first-$k$)      & 2.80 & \underline{4.41} & \underline{11.15} 
                      & 1.83 & \underline{2.35} & \underline{3.80}  &  +15.88\% \\
& \TheName{} (Ours)  & \textbf{2.88}   &  \textbf{4.51}  & \textbf{11.20}  & 
\textbf{1.88} &  \textbf{2.40}     &  \textbf{3.85}  & \textbf{+18.15\%}  \\
& Best Impr. &\cellcolor{lightgray}19.50\%&\cellcolor{lightgray}19.63\%&\cellcolor{lightgray}16.30\%&\cellcolor{lightgray}18.24\%&\cellcolor{lightgray}18.23\%&\cellcolor{lightgray}17.02\%& \cellcolor{lightgray}N/A\\
\bottomrule
\end{tabular}
\vspace{-2mm}
\caption{Performance comparison of various methods on the Amazon Books and Pixel datasets, using Qwen3-Embedding-0.6B as the backbone model, with text sequences truncated to 64 tokens.}  
\vspace{-6mm}
\label{tbl:avg_acc}
\end{table*}

\paragraph{Token Selection and Retraining}  
After optimizing the policy parameters, we proceed to select informative token chunks as follows. We define the probability of a chunk $c$ as  
$
P(c) = \prod_{i=1}^{|c|} p_i,
$
where $p_i$ denotes the generation probability of the $i$-th token in the chunk, and $|c|$ is the number of tokens in $c$. Chunks with higher probability scores according to $P(c)$ are selected. Using this set of high-probability chunks, we construct a refined dataset and retrain the model on this updated representation to improve both its performance and efficiency.

\section{Experiments}
In this section, we first describe the experimental setup, followed by a series of comprehensive experiments designed to evaluate the effectiveness and the efficiency of \TheName{}. 
Finally, we present a case study that illustrates the specific tokens selected by \TheName{}, providing insights into its operational mechanism.

\subsection{Experimental Setup}

We conduct a comprehensive evaluation of \TheName{} against four state-of-the-art baselines: 
(1) traditional ID-based sequential recommendation models, such as SASRec~\cite{kang2018selfattentive} and HSTU~\cite{hstu}; 
(2) semantic initialization methods that leverage large language model (LLM)-derived embeddings to initialize item ID representations, exemplified by LLMinit~\cite{Harte2023LeveragingLL}; 
and (3) hierarchical LLM-based frameworks, specifically HLLM~\cite{HLLM}. 
In addition, to study the impact of token selection, we adapt these baselines with different token selection policies, resulting in six baseline variants for comparison. We evaluate our approach on two publicly available datasets: the Amazon Book Reviews dataset~\cite{10.1145/2766462.2767755} and the Pixel dataset~\cite{cheng2023image}.  For backbone LLMs, we employ Qwen3-Embedding-0.6B~\cite{qwen3embedding} and TinyLlama-1.1B~\cite{Zhang2024TinyLlamaAO}. Performance is evaluated using Recall@K (R@K) and NDCG@K (N@K). Additional details on the experimental setup are given in Appendix~\ref{app:sec:more-setup}.

\subsection{Experimental Results}
\paragraph{Overall Performance}
As shown in Table~\ref{tbl:avg_acc}, \TheName{} significantly outperforms all baseline approaches by a substantial margin. In this setting, input sequences are truncated to 64 tokens, and the Qwen3-Embedding-0.6B model serves as the frozen LLM backbone. For \TheName{}, the \textit{chunk size} is set to 8.
On the Amazon Books dataset, \TheName{} improves upon SASRec by up to 29.29\% in Recall@K and 31.25\% in NDCG@K, yielding an average gain of 29.43\% across metrics. On the Pixel dataset, it achieves improvements of up to 19.63\% in Recall@K and 18.24\% in NDCG@K, with an average increase of 16.79\%.

\paragraph{The Importance of Tokens in LLM-based Recommendation}
Given the significant advances achieved by LLM-based sequential recommendation methods, we investigate the impact of different types of input tokens. Specifically, we compare \textit{random tokens}, \textit{top-$k$ logit tokens}, \textit{the first $k$ tokens}, and \TheName{}. For \textit{random tokens}, we randomly select $k$ tokens from the item text. For \textit{top-$k$ logit tokens}, we select the $k$ tokens with the highest logits. \textit{The first $k$ tokens} is the default strategy employed by many existing methods.
In contrast, \TheName{} selects $k$ tokens using a learned policy.

\begin{figure*}[t]
  \centering
    \begin{subfigure}[b]{0.24\textwidth}
        \centering
        \includegraphics[width=\textwidth]{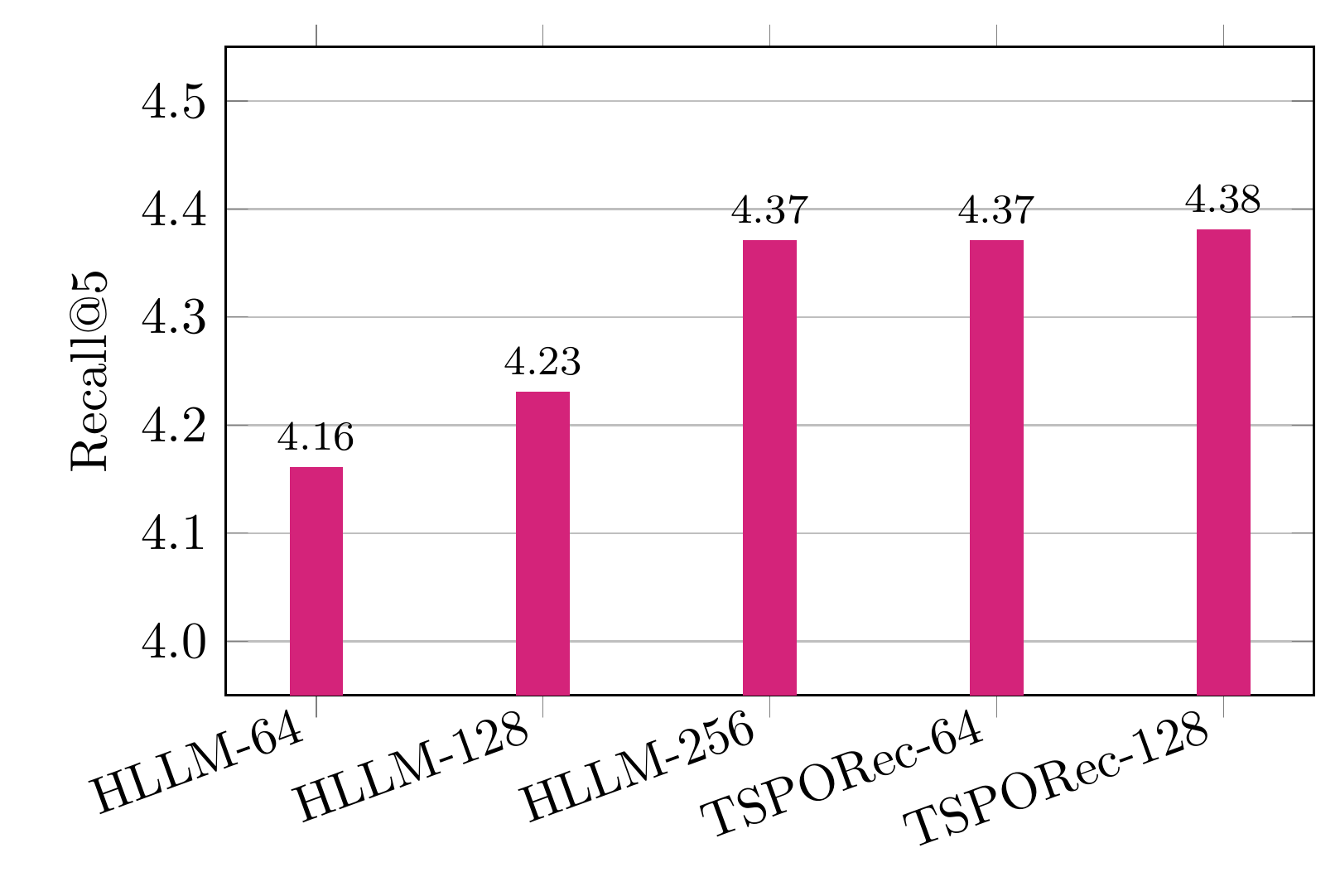} 
        \caption{Amazon Book, Recall@5}
        \label{fig:sub1}
    \end{subfigure}
    \hfill 
    \begin{subfigure}[b]{0.24\textwidth}
        \centering
        \includegraphics[width=\textwidth]{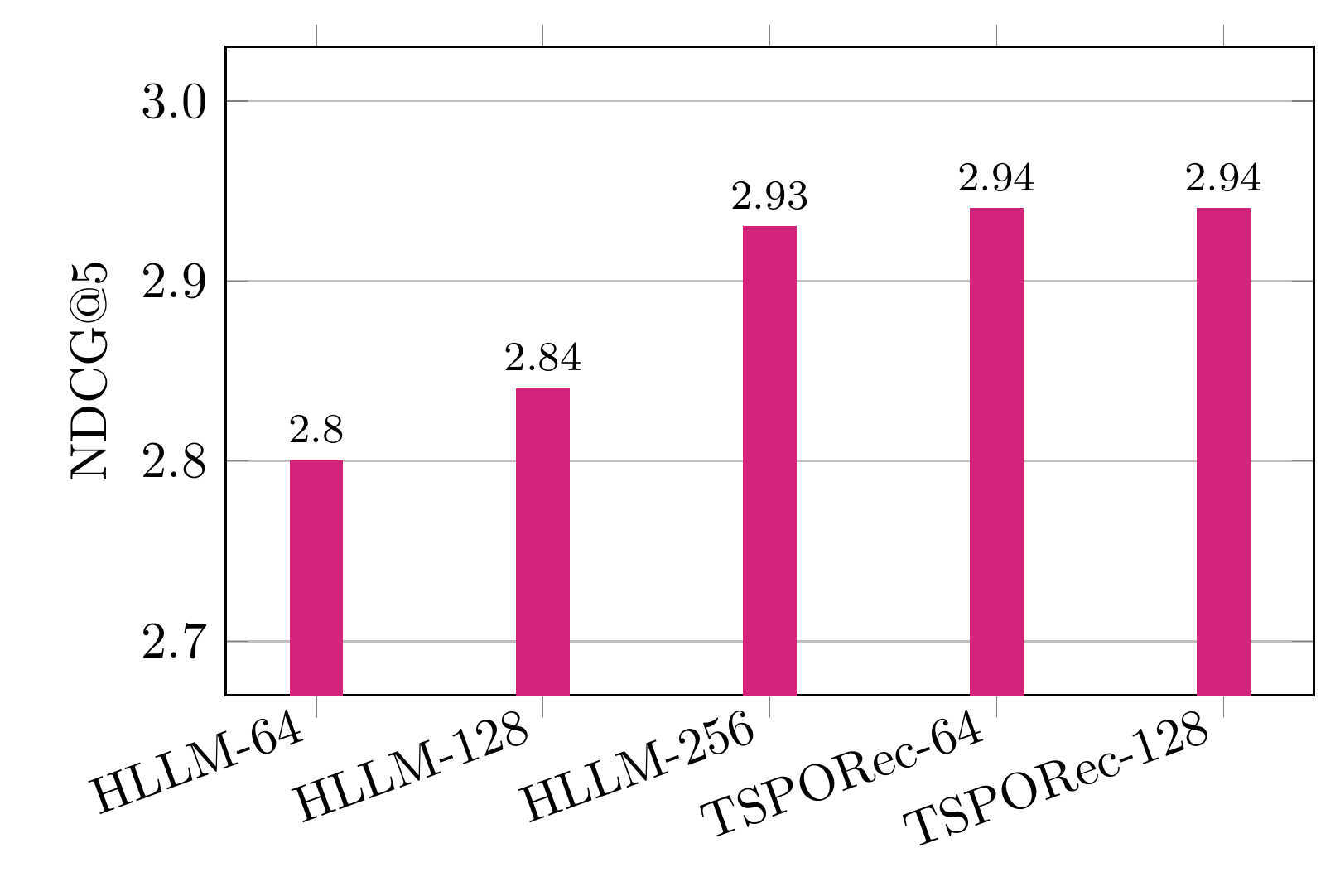} 
        \caption{Amazon Book, NDCG@5}
        \label{fig:sub1}
    \end{subfigure}
    \hfill 
    \begin{subfigure}[b]{0.24\textwidth}
        \centering
        \includegraphics[width=\textwidth]{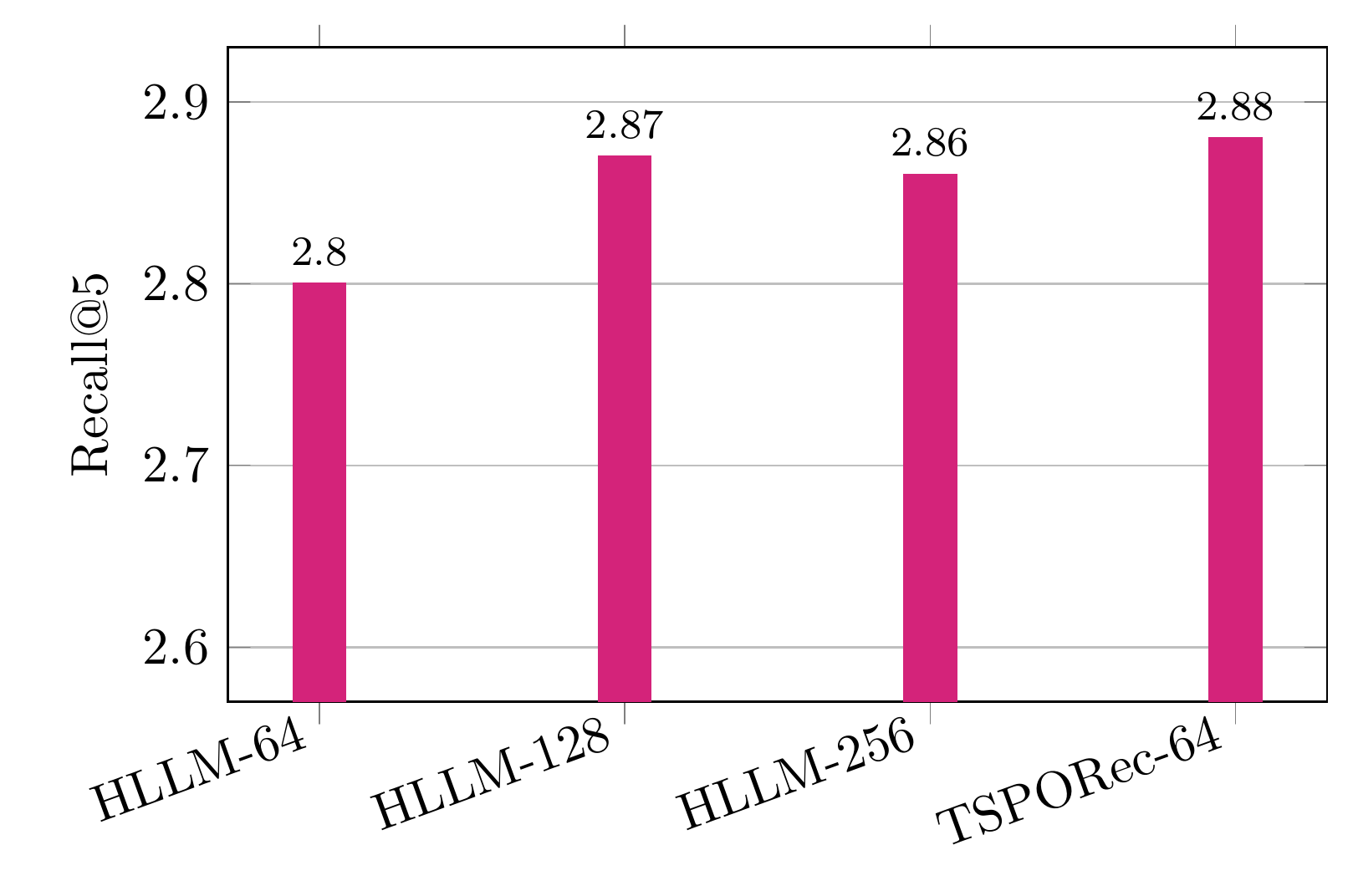} 
        \caption{Pixel, Recall@5}
        \label{fig:sub1}
    \end{subfigure}
    \hfill 
    \begin{subfigure}[b]{0.24\textwidth}
        \centering
        \includegraphics[width=\textwidth]{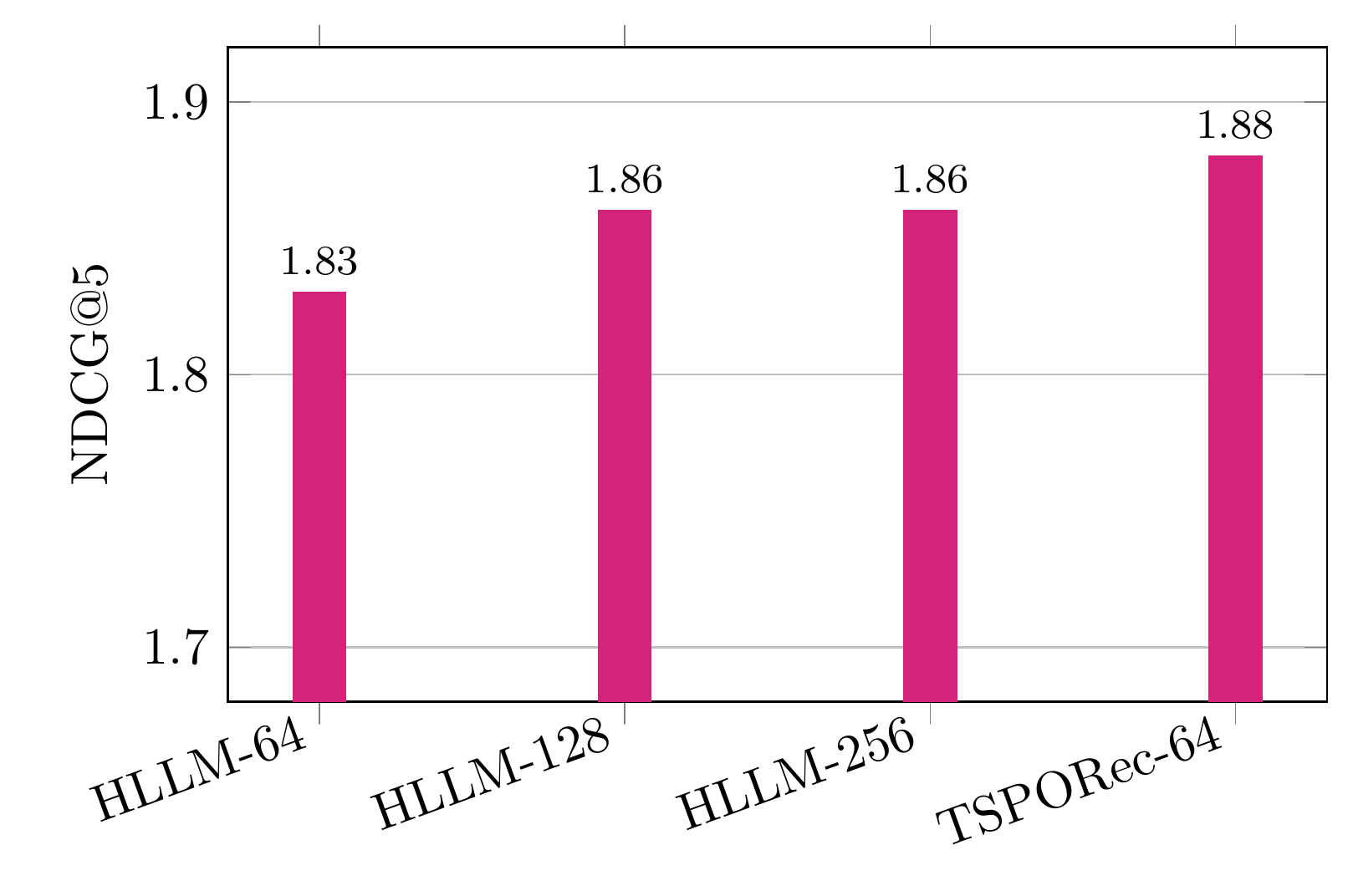} 
        \caption{Pixel, NDCG@5}
        \label{fig:sub1}
    \end{subfigure}
    \begin{subfigure}[b]{0.24\textwidth}
        \centering
        \includegraphics[width=\textwidth]{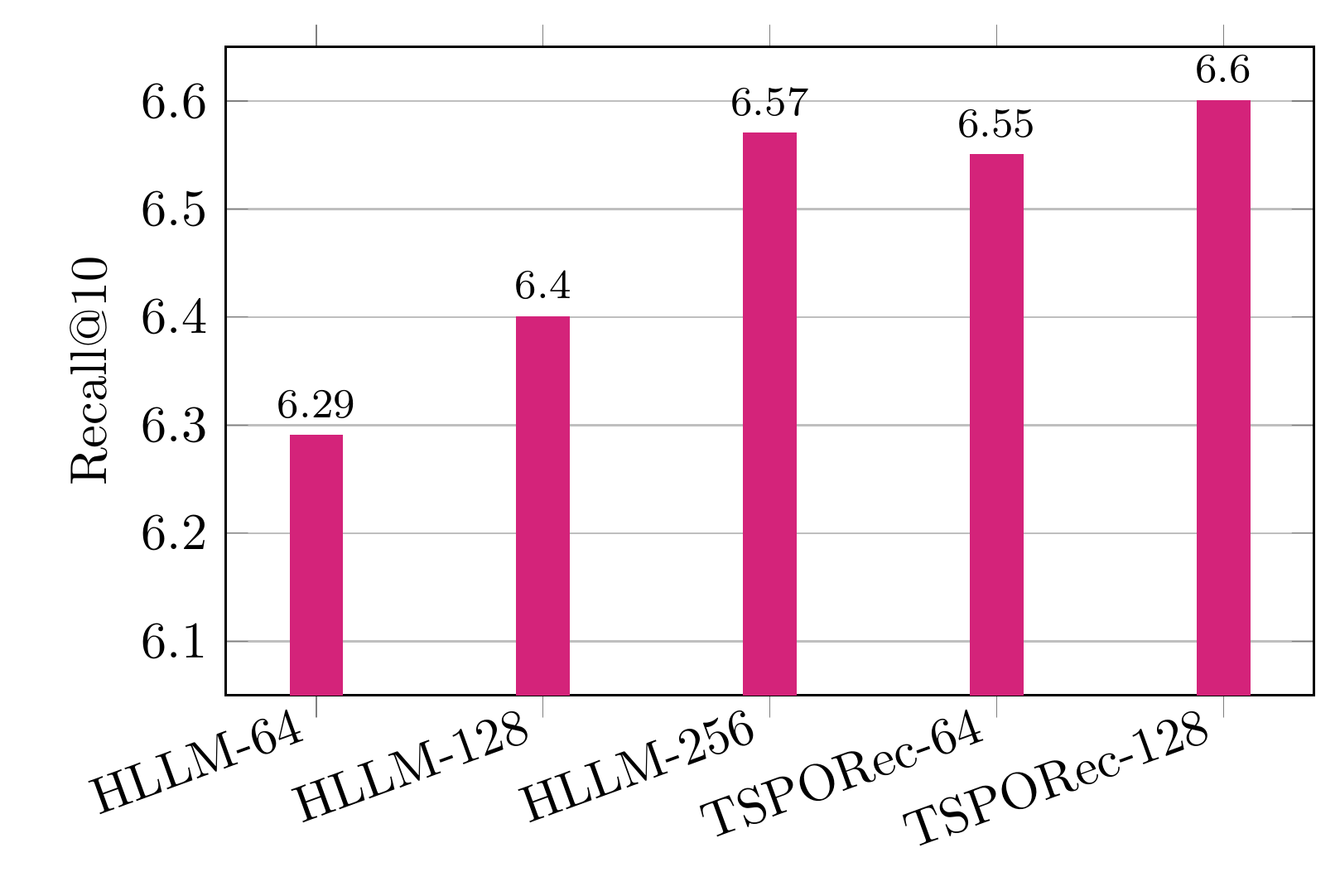} 
        \caption{Amazon Book, Recall@10}
        \label{fig:sub1}
    \end{subfigure}
    \hfill 
    \begin{subfigure}[b]{0.24\textwidth}
        \centering
        \includegraphics[width=\textwidth]{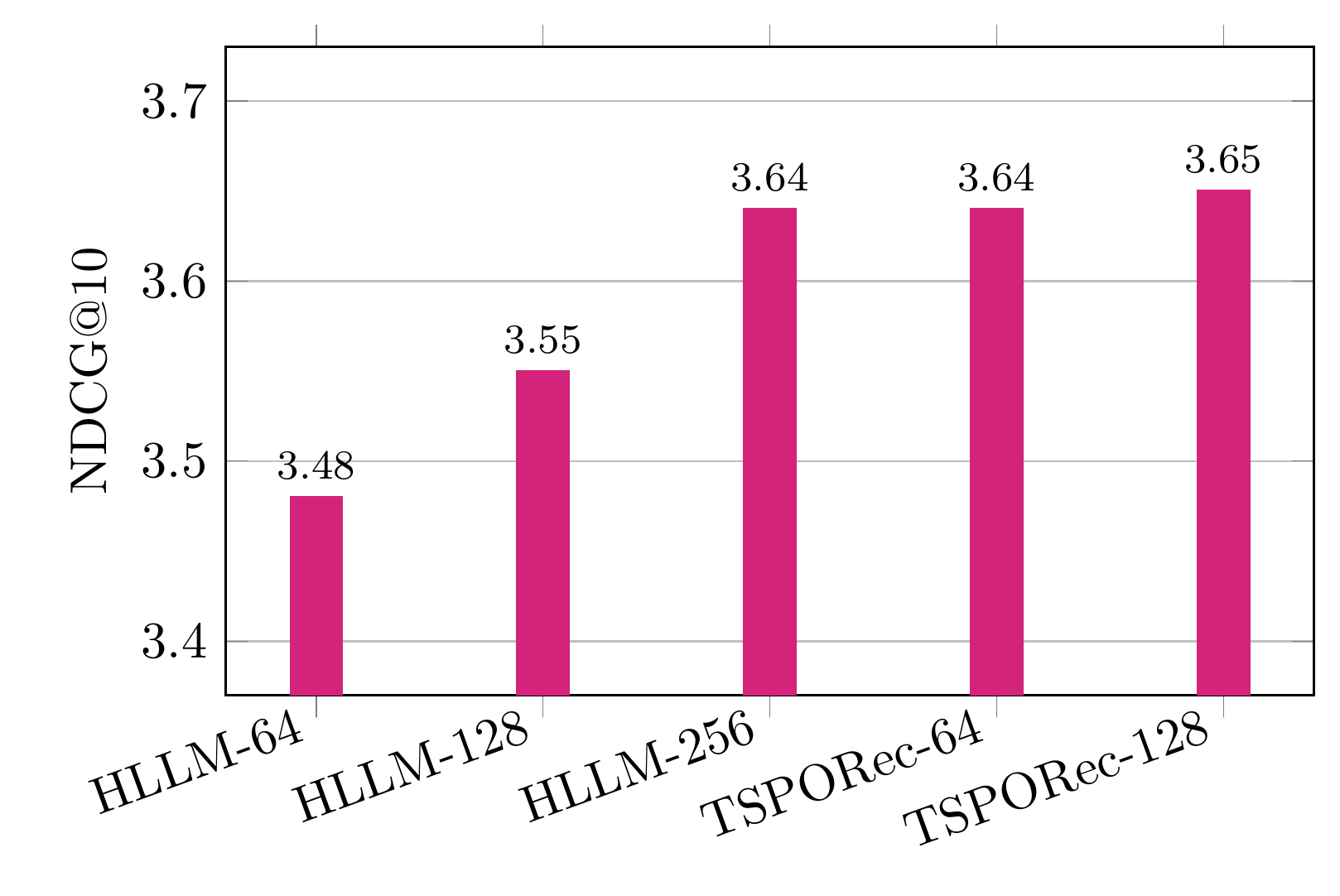} 
        \caption{Amazon Book, NDCG@10}
        \label{fig:sub1}
    \end{subfigure}
    \hfill 
    \begin{subfigure}[b]{0.24\textwidth}
        \centering
        \includegraphics[width=\textwidth]{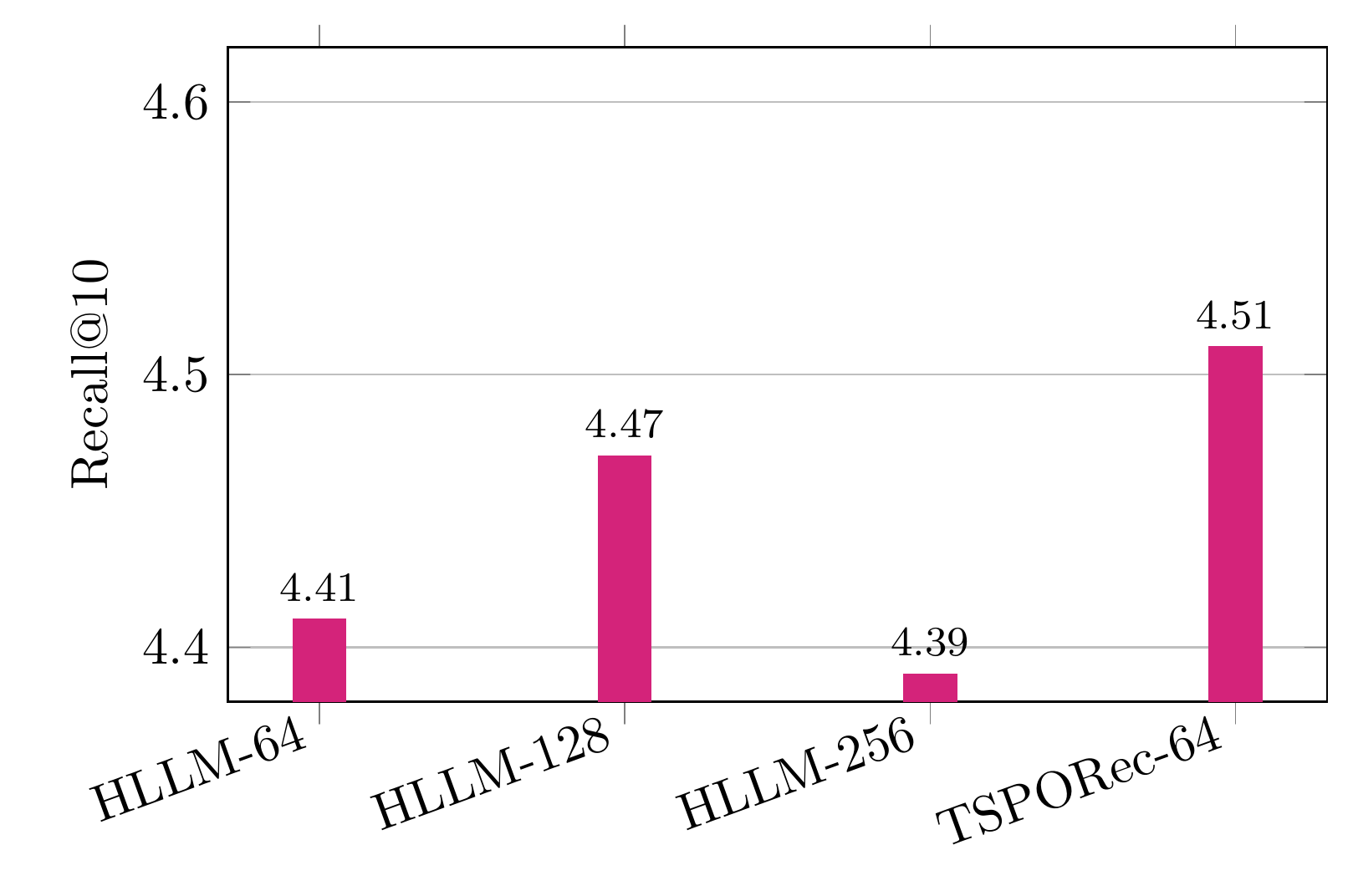} 
        \caption{Pixel, Recall@10}
        \label{fig:sub1}
    \end{subfigure}
    \hfill 
    \begin{subfigure}[b]{0.24\textwidth}
        \centering
        \includegraphics[width=\textwidth]{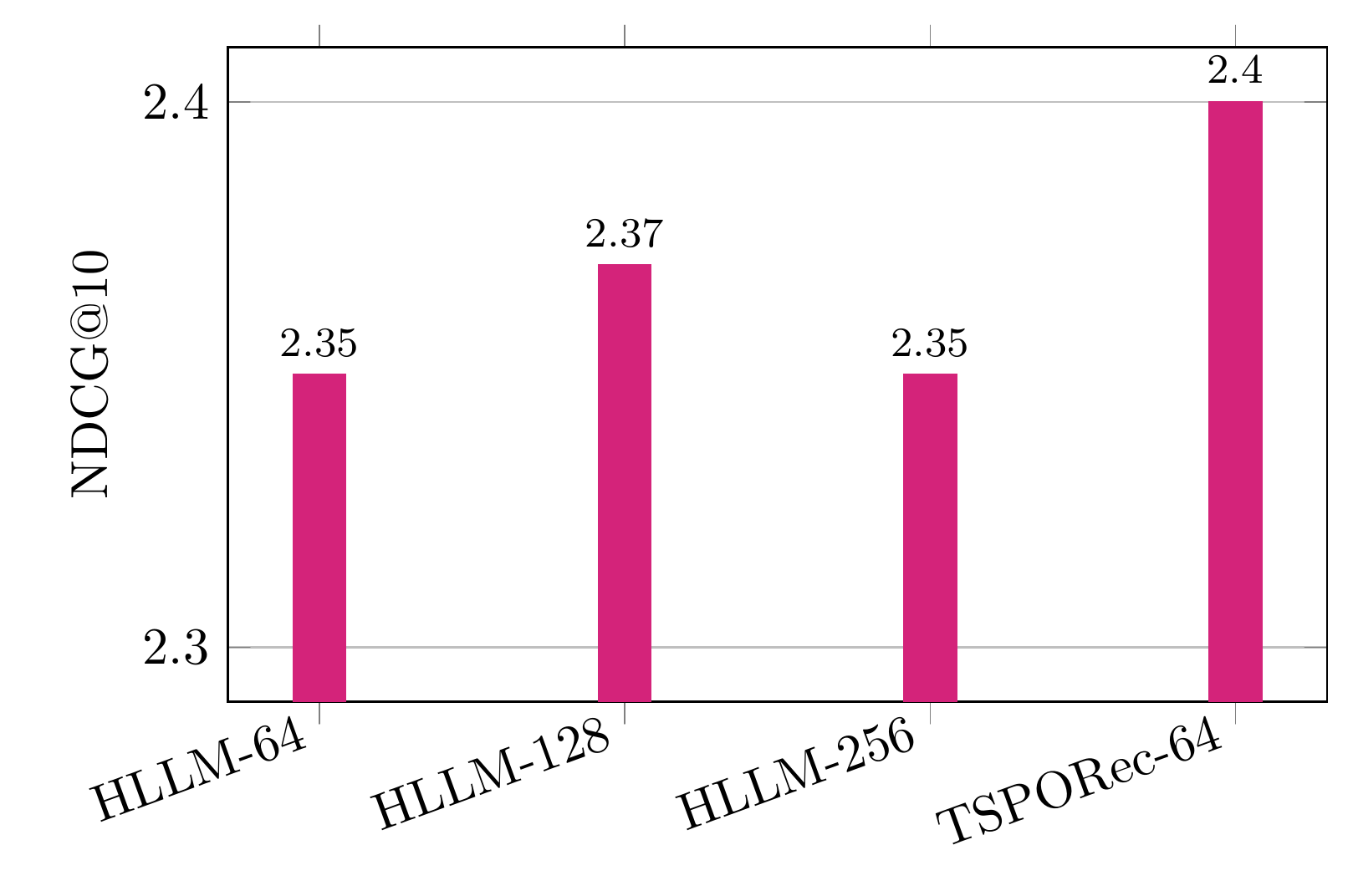} 
        \caption{Pixel, NDCG@10}
        \label{fig:sub1}
    \end{subfigure}   
\vspace{-2mm}
\caption{Performance of HLLM and \TheName{} on Amazon Books and Pixel with different token sequence lengths, using \texttt{Qwen3-Embedding-0.6B} as the backbone model.}   
\vspace{-6mm}
\label{fig:vary_len_cmp}
\end{figure*}

As shown in Table~\ref{tbl:avg_acc}, \TheName{} achieves the most significant performance gain by leveraging the learned token selection policy, substantially outperforming existing token selection strategies. On the Amazon Books dataset, randomly selecting tokens reduces the performance improvement of HLLM from 24.40\% to just 2.14\%, highlighting the critical role of input token selection in LLM-based recommendation.

These results demonstrate that the choice of input tokens has a profound impact on model effectiveness—a factor largely overlooked by prior approaches. 
Surprisingly, the \textit{top-$k$ logit tokens} do not outperform the simple \textit{first-$k$ tokens} baseline, suggesting that more principled token selection mechanisms are needed. 
This limitation may stem from the unidirectional nature of standard LLM attention, which prevents hidden states from effectively capturing meaningful collaborative signals across the sequence. In contrast, as illustrated in Eq.~\eqref{eq:logits}, \TheName{} leverages a bidirectional scoring mechanism. The importance score of each token is learned from both (i) the standard attention mechanism, which encodes semantic meaning, and (ii) the pretrained item embedding $\mathbf{h}_{k+1}$, which provides a collaborative filtering signal. 
This dual-source design enables \TheName{} to jointly capture rich semantic information and sequential collaborative patterns, thereby facilitating the identification of the most informative tokens and ultimately improving model performance.

\paragraph{The Impact of Token Sequence Length in LLM-based Recommendation}

In this section, we further investigate how the length of input tokens impacts recommendation performance. To this end, we conduct extensive experiments with varying input token lengths, as shown in Figure~\ref{fig:vary_len_cmp}. We denote each configuration as \textit{Method}$-k$, where $k$ represents the number of input tokens used. For instance, HLLM-$128$ indicates the use of the first 128 tokens from the input sequence, while \TheName{}-$64$ denotes the selection of 64 tokens according to the learned token selection policy.

As illustrated in Figure~\ref{fig:vary_len_cmp}, on the Amazon Books dataset, increasing the token length leads to performance improvements for both HLLM and \TheName{}, suggesting that longer input sequences may contain additional informative tokens that benefit recommendation quality. In contrast, on the Pixel dataset, while moderate increases in token length initially enhance HLLM's performance, further lengthening the input sequence eventually degrades performance. This indicates that not all tokens contribute positively to recommendations, and including irrelevant or noisy tokens can be detrimental. This phenomenon further underscores the necessity of effective token selection to improve recommendation performance. Moreover, on both the Amazon Books and Pixel datasets, \TheName{} with only 64 selected tokens achieves performance on par with or superior to the corresponding HLLM variants using 256 tokens. These results demonstrate that \TheName{} not only enhances recommendation accuracy but also improves computational efficiency by enabling effective performance with significantly fewer tokens.

\paragraph{The Efficiency of \TheName{}}

To evaluate the impact on inference efficiency of \TheName{}, we analyze the computational cost across varying input sequence lengths. 
By selecting informative tokens, \TheName{} effectively reduces the required input length without sacrificing performance, thereby lowering overall computational overhead.

We measure the inference time of LLM on an NVIDIA H100 GPU using the Amazon Books dataset, with a batch size of 32 and an item sequence length of 10. 
The results are summarized in Table~\ref{tbl:infer_time_cost_hllm}. 
We observe that the item LLM dominates the inference time, and its computational share grows significantly as the input token sequence length increases. 
Since \TheName{} matches the performance of HLLM using only 64 tokens with HLLM using 256 tokens, it reduces inference cost by \textbf{63.4\%} in scenarios where item embeddings are precomputed offline, and by \textbf{61.3\%} when both the item and user LLMs are served online.

\begin{table}[t]
\centering
\small
\begin{tabular}{c cc cc c}
\toprule
Token$_S$ & I$_t$ (ms)  & U$_t$ (ms)  & Total (ms)  \\
\midrule
64  & 299  & 27  & 326\\
128 & 509  & 27  & 536\\
256 & 817  & 26  & 843\\
\bottomrule
\end{tabular}
\vspace{-2mm}
\caption{Inference time cost during evaluation. 
$\text{Token}_S$ denotes the length of the input token sequence, $I_t$ the inference time of the item LLM, and $U_t$ the inference time of the user LLM, using Qwen3-Embedding-0.6B.}
\vspace{-6mm}
\label{tbl:infer_time_cost_hllm}
\end{table}

\paragraph{The Generalization capability of \TheName{}}

To evaluate the generalization ability of \TheName{}, we conduct experiments using a different LLM backbone---\texttt{TinyLlama-1.1B}, as well as a different downstream SR model, LLMinit~\cite{Harte2023LeveragingLL}. We compare \TheName{} against the HLLM variant with the same LLM backbone, and the results are presented in Table~\ref{tbl:tiny_llama_amz}. \TheName{} consistently outperforms HLLM across all Recall and NDCG metrics, achieving an average improvement of 4.12\%. Furthermore, as shown in Table~\ref{tbl:llminit_toks}, when \TheName{}-selected tokens are used to initialize LLMinit, the model performance improves by 3.76\%, demonstrating that the selected tokens carry informative signals for recommendation. These results confirm that \TheName{} not only enhances downstream performance but also exhibits strong generalization across different LLM architectures and recommendation models.

\subsection{Case Study}
\begin{figure*}[t]
\centering
\includegraphics[width=\linewidth]{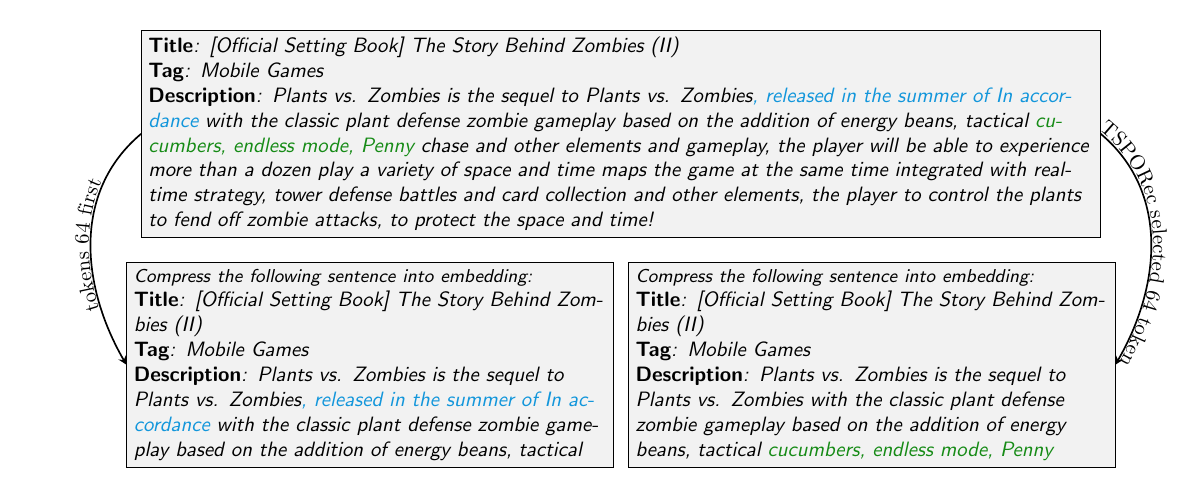}
\vspace{-8mm}
\caption{Case study on token selection for item ``i9019'' in the Pixel dataset, with 64 tokens selected by each method.}
\vspace{-4mm}
\label{fig:case_1}
\end{figure*}

\begin{figure*}[t]
\centering
\includegraphics[width=\linewidth]{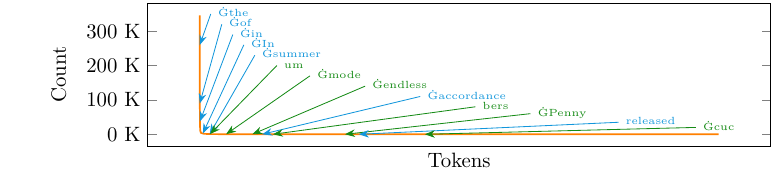}
\vspace{-8mm}
\caption{Comparison of token distributions selected by first-$k$ and \TheName{} ($\dot{G}$ is a leading space token).}
\vspace{-4mm}
\label{fig:case_1_tok_dist}
\end{figure*}

In this section, we conduct a case study to examine the underlying mechanisms of token selection. An illustrative example is shown in Figure~\ref{fig:case_1}, with additional cases provided in Appendix~\ref{app:sec:case2} for further analysis.
We compare the tokens selected by the first-$k$ strategy and \TheName{} on the Pixel dataset. As depicted in Figure~\ref{fig:case_1}, notable differences exist between the two strategies, highlighted in green (tokens selected by \TheName{} but not by first-$k$) and blue (tokens selected by first-$k$ but not by \TheName{}). 

\begin{table}[t]
\small
\begin{tabular}{c cccc r}
\toprule
Method &  R@5 & R@10 & N@5 & N@10  & Impr. \\
\midrule
HLLM  & 3.99 & 6.07& 2.66&  3.33 & +0.0\%\\
\TheName{} & \textbf{4.16}& \textbf{6.30}& \textbf{2.78}& \textbf{3.46}& \textbf{+4.12\%}\\
\bottomrule
\end{tabular}
\vspace{-2mm}
\caption{Performance comparison of HLLM and \TheName{} on the Amazon Books dataset, using TinyLlama-1.1B as the backbone model, with text sequences truncated to 64 tokens.}  
\vspace{-4mm}
\label{tbl:tiny_llama_amz}
\end{table}

\begin{table}[t]
\small
\centering
\begin{tabular}{c cccc r}
\toprule
Method &  R@5 & R@10 & Impr. \\
\midrule
LLMinit(first-$k$) & 3.29 & 5.05 & + 0.0\%\\
LLMinit(\TheName{}) & \textbf{3.42} & \textbf{5.23} & \textbf{+3.76\%}\\
\bottomrule
\end{tabular}
\vspace{-2mm}
\caption{Performance comparison of LLMinit on the Amazon Books dataset, with different tokens for LLM embedding generation.}  
\vspace{-8mm}
  \label{tbl:llminit_toks}
\end{table}

We observe that \TheName{} removes the text segment \textcolor{sftblue}{\textit{``released in the summer of In accordance''}} and instead selects new content such as \textcolor{mygreen}{\textit{``cucumber, endless mode, Penny''}}. 
The newly selected tokens are predominantly content words—nouns, adjectives, or named entities—that convey specific and semantically rich information about the item. 
In contrast, the removed tokens consist largely of function words (e.g., prepositions, pronouns, and articles) that carry limited discriminative power. 
By prioritizing content words, \TheName{} enriches item representations with more informative features, thereby enhancing downstream recommendation performance.

Figure~\ref{fig:case_1_tok_dist} further illustrates that \TheName{} tends to filter out frequent (``hot'') tokens. This behavior is meaningful: hot tokens are shared across many items and thus exhibit low semantic specificity, making it difficult for the user LLM to distinguish between items and learn effective representations.

\section{Related Work}

\paragraph{LLM-based SR} LLMs possess extensive world knowledge and strong reasoning capabilities, offering two principal advantages for sequential recommendation: rich semantic understanding and powerful modeling capacity derived from their high-dimensional parameter space. These characteristics enable LLMs to effectively capture complex user behavior patterns and thereby enhance recommendation performance \cite{sun2024large, zhang2024notellm, zhang2025notellm, qiao2024llm4sbr, liu2024llm, xu2024slmrec}. In the context of SR, LLM-based approaches can be broadly categorized into two groups. 

The first category leverages LLMs as semantic initializers to enrich ID-based item embeddings
\cite{Harte2023LeveragingLL, liu2024large, alphafuse}. A key challenge in these methods arises from the dimensional mismatch between the high-dimensional LLM embeddings and the typically lower-dimensional ID embeddings used in recommendation systems. To address this, dimensionality reduction techniques are commonly employed—such as Principal Component Analysis in LLMEmb \cite{liu2024large} and Singular Value Decomposition in AlphaFuse \cite{alphafuse}.

The second category adopts LLMs in an end-to-end manner, where the model directly utilizes the LLM and tokenized item sequences to predict the next item in the interaction sequence of users. HLLM \cite{HLLM}
demonstrates significant performance improvements over traditional methods. This end-to-end integration enables modeling of both semantic content 
and collaborative signals from user interactions, thereby advancing the recommendation performance.

However, due to computational constraints, only the first portion of the textual description of each item is used.
This truncation limits the model's ability to capture the full richness of the semantic information present in longer descriptions, potentially leading to suboptimal performance. In contrast, our method selects the most informative tokens from item descriptions.

\section{Conclusion}
In this paper, we propose a novel Token Selection approach tailored for LLM-driven Sequential Recommendation systems, i.e., \TheName{}. 
We develop a three-stage workflow. First, we pre-train the foundational LLM-based sequential recommendation model to establish a robust baseline. Second, we train a dedicated policy to identify informative tokens; additionally, a proxy reward function is designed to facilitate chunk-oriented token selection, addressing the challenges of granularity in token identification. Finally, leveraging the identified tokens, we retrain the pre-trained baseline model to optimize its recommendation performance. Comprehensive experimental results demonstrate the effectiveness (up to 31.25\%) and efficiency (up to 63.4\%) of the proposed method.

\cleardoublepage
\section*{Limitations}
The proposed \TheName{} employs a three-stage pipeline to select the most informative tokens from the item’s textual features.
While \TheName{} significantly improves both recommendation performance and inference efficiency compared to standard training methods, it incurs additional training time due to Stages 2 and 3. However, this overhead is practically acceptable, as models are typically trained once and deployed over an extended period—rendering inference efficiency a more critical concern than the one-time training cost.
\bibliography{acl}

\appendix
\input{appendix.tex}

\end{document}

%% file: appendix.tex
\clearpage
\setcounter{theorem}{0}
\section{More Experimental Setup}
\label{app:sec:more-setup}
\begin{table}[h]
\centering
\small
\begin{tabular}{l c c c}
  \toprule
  Dataset      & \#Users   & \#Items   & \#Interactions \\
  \midrule
  Amazon Books & 694,898   & 686,624   & 10,053,086     \\
  Pixel        & 200,000   & 96,282    & 3,965,656      \\
  \bottomrule
\end{tabular}
\caption{Statistics of the Pixel and Amazon Books Datasets.}
\label{app:tbl:stat}
\end{table}

In this section, we detail the experimental setup used to evaluate our approach.

We conduct experiments on two publicly available datasets: the Amazon Book Reviews dataset~\cite{10.1145/2766462.2767755} and the Pixel dataset~\cite{cheng2023image}. We present the statistics of the datasets used in Table~\ref{app:tbl:stat}. 
The Amazon Books dataset contains 694,898 users and 686,624 items, while the Pixel dataset contains 200,000 users and 96,282 items. Both datasets exhibit high sparsity, with user-item interaction densities of $2.1 \times 10^{-5}$ and $2.1 \times 10^{-4}$, respectively. Items with missing attributes are filtered out. We adopt a leave-one-out evaluation protocol: the most recent interaction is used for testing, the second-to-last for validation, and all earlier interactions for training. Performance is evaluated using Recall@K (R@K) and NDCG@K (N@K).

Following the setup of HLLM~\cite{HLLM}, we set the learning rate to $1 \times 10^{-4}$ for all baseline models. HLLM is trained for 5 epochs, while other models are trained for up to 200 epochs with early stopping to prevent overfitting. We fix the number of negative samples at 128 and set the maximum sequence length to 10. In the policy learning phase, the model is trained for 5 epochs, with a learning rate $1 \times 10^{-3}$. 

To investigate the impact of textual input length, we truncate item descriptions to 64, 128, and 256 tokens, respectively. For the Amazon Books dataset, we use the \textit{title} and \textit{description} fields. For the Pixel dataset, we utilize the \textit{title}, \textit{tag}, and \textit{description} text features.
Since starting tokens occupy a distinct feature space~\cite{han2024lm}, \TheName{} retains the first \textit{prefix size} to 16 tokens from all sampled sequences.

We train all models on a GPU cluster equipped with 8 H100 and 8 A800 GPUs. Each experiment is conducted three times, and we report the average performance to ensure statistical reliability.

\section{Case Study}
\label{app:sec:case2}
\begin{figure*}[t]
\centering
\includegraphics[width=\linewidth]{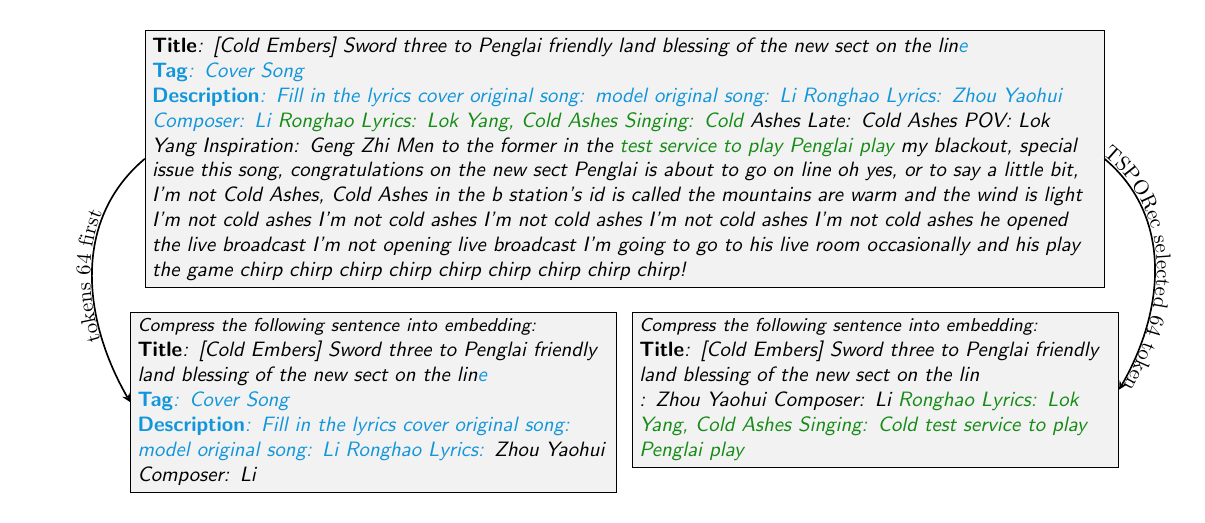}
\caption{Case study on token selection for item ``i113948'' in the Pixel dataset, with 64 tokens selected by each method.}
\label{fig:case_2}
\end{figure*}
We conduct an additional case study to further examine the token selection behavior of the competing methods. The results for item \textit{``i113948''} are illustrated in Figure~\ref{fig:case_2}. Consistent with earlier observations, significant differences are evident between the two strategies: tokens selected exclusively by \TheName{} are highlighted in green, while those selected only by the first-$k$ strategy are marked in blue, allowing for a fine-grained comparison of their selection patterns.

\begin{figure*}[t]
\centering
\includegraphics[width=\linewidth]{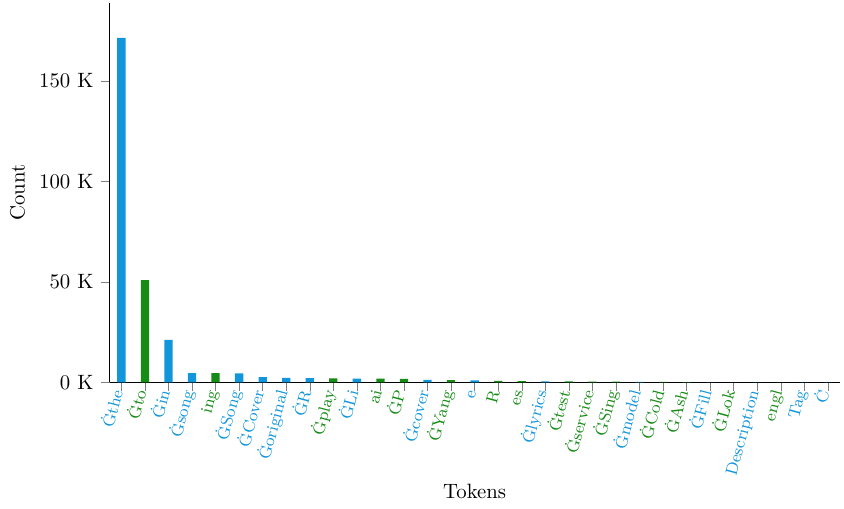}
\caption{Comparison of token distributions selected by first-$k$ and \TheName{}.}
\label{app:fig:case_2_token_dist}
\end{figure*}

In this instance, \TheName{} removes tag-related segments such as \textcolor{sftblue}{\textit{``tag: Cover Song''}} and \textcolor{sftblue}{\textit{``Fill in the lyrics cover original song: model original song''}}, and instead selects the segment \textcolor{mygreen}{\textit{``Ronghao Lyrics: Lok Yang, Cold Ashes Singing: Cold test service to play Penglai play''}}. 
This selection is semantically meaningful: for a music item, metadata about the lyricist and composer provides more concrete, fine-grained, and discriminative signals than generic categorical tags. 
By prioritizing such specific metadata, \TheName{} enriches the item representation, thereby improving downstream recommendation performance.

As shown in Figure~\ref{app:fig:case_2_token_dist}, consistent with the previous case study, \TheName{} tends to filter out high-frequency tokens. 
This behavior encourages the item LLM to learn more discriminative embeddings, as frequent tokens often carry less item-specific semantic information. Please note that, whereas Figure~\ref{fig:case_1_tok_dist} shows the token distribution across the full set of text features, Figure~\ref{app:fig:case_2_token_dist} depicts the distribution of token counts in the selection results. The results reveal that—under both counting metrics—the proposed method consistently filters out frequent (“hot”) tokens.

\section{More Experimental Results}
\label{app:sec:more-exp}
\paragraph{The Impact of Training and Selection Chunk Sizes}
\begin{figure*}[t]
  \centering
    \begin{subfigure}[b]{0.24\textwidth}
        \centering
        \includegraphics[width=\textwidth]{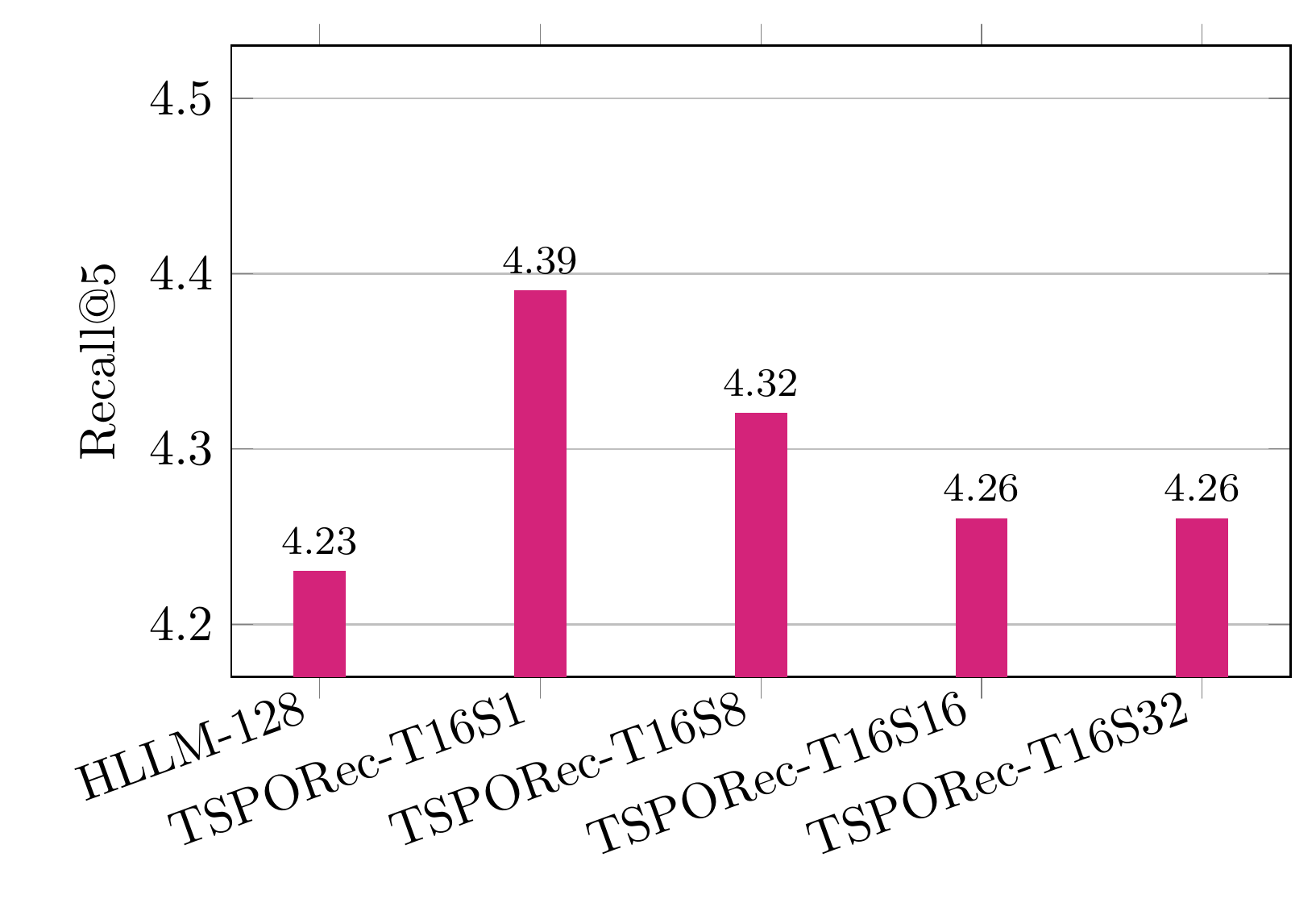} 
        \caption{Recall@5}
        \label{fig:sub1}
    \end{subfigure}
    \hfill 
    \begin{subfigure}[b]{0.24\textwidth}
        \centering
        \includegraphics[width=\textwidth]{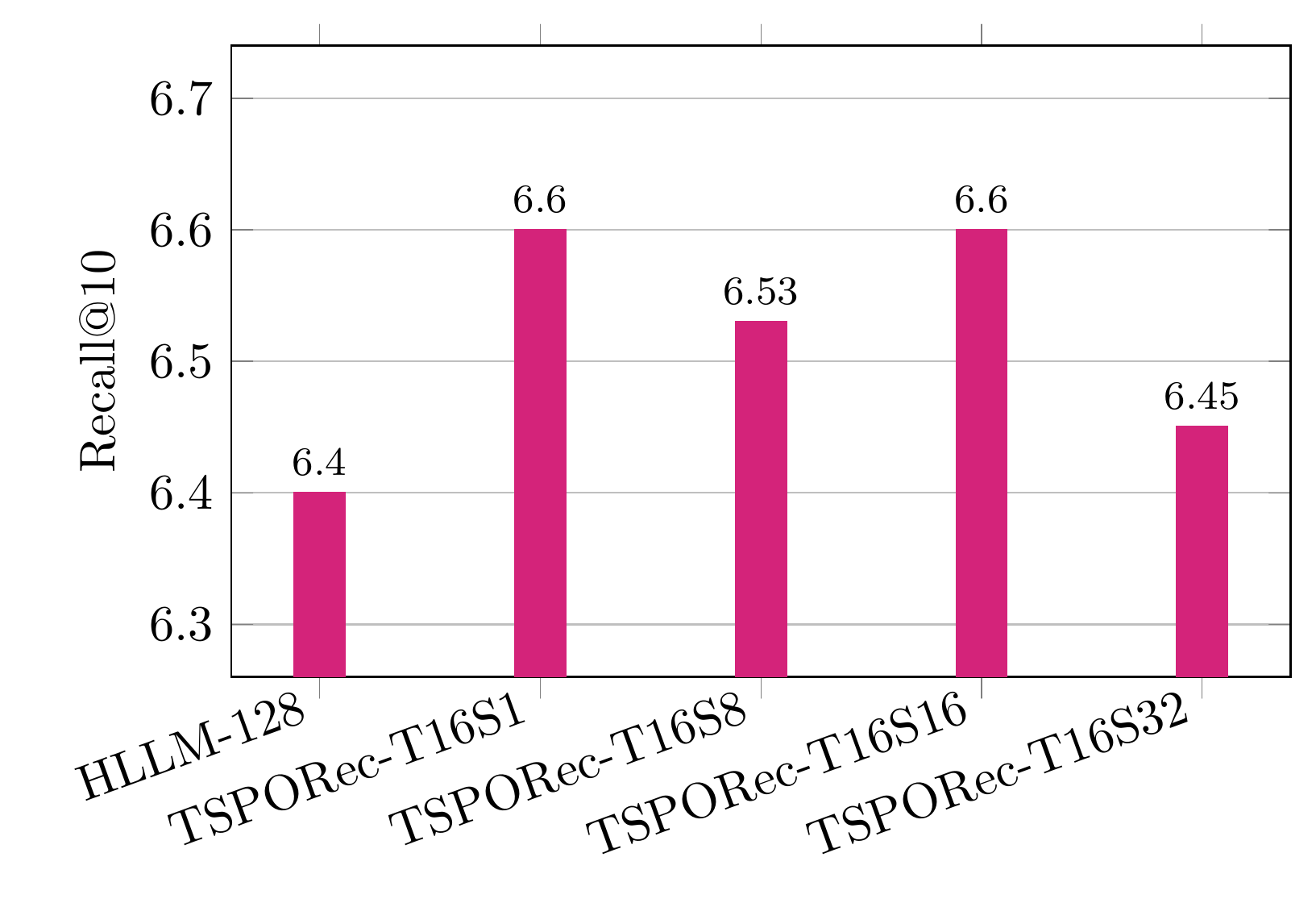} 
        \caption{Recall@10}
        \label{fig:sub1}
    \end{subfigure}
    \hfill 
    \begin{subfigure}[b]{0.24\textwidth}
        \centering
        \includegraphics[width=\textwidth]{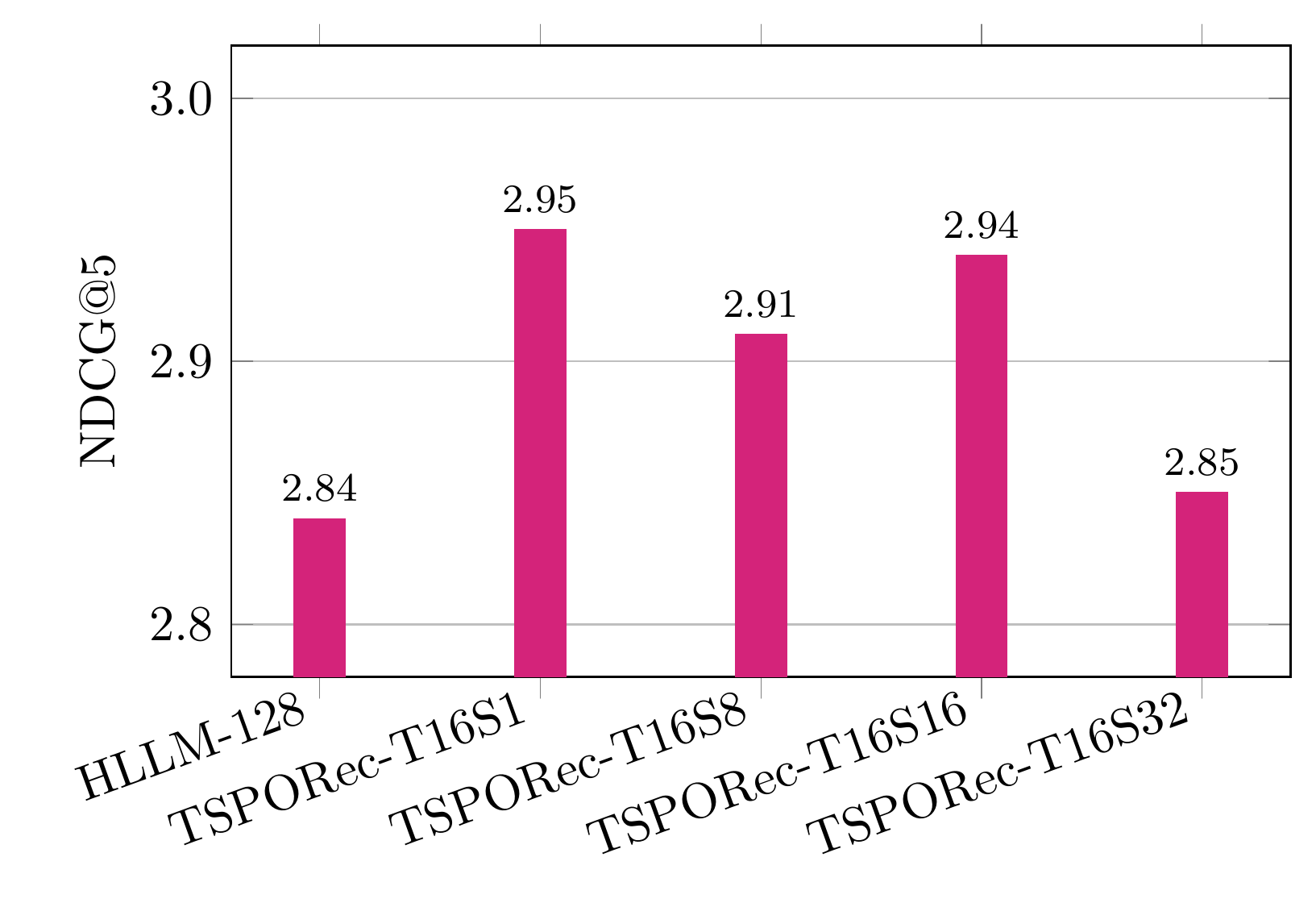} 
        \caption{NDCG@5}
        \label{fig:sub1}
    \end{subfigure}
    \hfill 
    \begin{subfigure}[b]{0.24\textwidth}
        \centering
        \includegraphics[width=\textwidth]{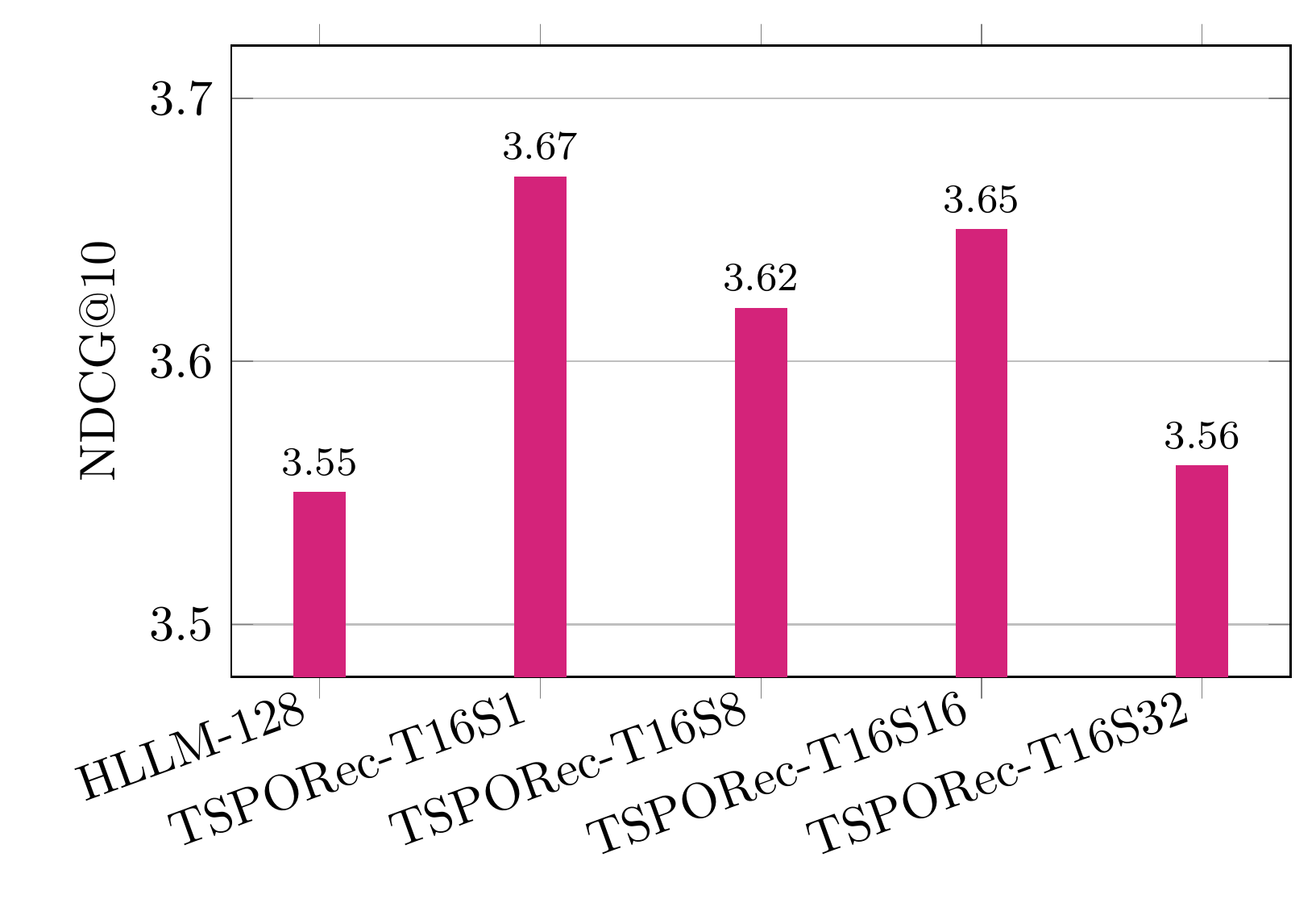} 
        \caption{NDCG@10}
        \label{fig:sub1}
    \end{subfigure}
\caption{Performance of HLLM and \TheName{} on Amazon Books with different training and selection chunk sizes, using \texttt{Qwen3-Embedding-0.6B} as the backbone model.}    
\label{fig:vary_chunk_size}
\end{figure*}

Since the training and selection chunk sizes may differ, we conduct an ablation study to investigate the impact of training and selection chunk sizes on recommendation performance.  we denote each configuration as \TheName{}-T$k_1$S$k_2$, where the model is trained with chunk size $k_1$ and performs token selection during inference with chunk size $k_2$. We evaluate \TheName{} on the Amazon Books dataset using an input sequence length of 128 tokens. The results, shown in Figure~\ref{fig:vary_chunk_size}, demonstrate that: (1) across all configurations, the proposed \TheName{} consistently outperforms the HLLM baseline by up to 3.54\%, indicating robustness to varying parameter settings.

(2) On average, \TheName{}-T16S1, \TheName{}-T16S8, \TheName{}-T16S16, and \TheName{}-T16S32 outperform HLLM by 3.54\%, 2.15\%, 2.54\%, and 0.53\%, respectively. A general downward trend in performance is observed as the selection chunk size increases, indicating that smaller chunks yield more effective token selection. We further validate this trend using top-$k$ logit-based selection, with results presented in Table~\ref{app:tbl:topk_logits_selection}. The consistent performance degradation with larger chunk sizes across both methods suggests a shared preference for finer-grained selection units. This observation further implies that the limited effectiveness of top-$k$ logit selection may stem from its unidirectional nature, rather than from disruptions to semantic sentence structure.

\paragraph{Training Time}
Since \TheName{} employs a three-stage training pipeline, we report the training time for each stage in Table~\ref{app:tbl:training-time}. 
The model is trained on a GPU cluster equipped with 8 NVIDIA H100 GPUs. 
Due to the additional stages, \TheName{} requires approximately 21.5 hours more than a single-stage baseline. 
This overhead is acceptable in practice, as the full training procedure is performed only once and the resulting model is deployed for an extended period.

\begin{table*}[t]
  \centering
  \begin{tabular}{ccc}
    \toprule
    Stage 1: \textit{Pretraining} & 
    Stage 2: \textit{Token Selection Policy Learning} & 
    Stage 3: \textit{Token Selection and Retraining} \\
    \midrule
    13h\,49m\,1s & 8h\,56m\,24s & 12h\,40m\,52s \\
    \bottomrule
  \end{tabular}
  \caption{Training time of \TheName{} for each stage.}
  \label{app:tbl:training-time}
\end{table*}

\begin{table}[t]
  \centering
  \small
\begin{tabular}{c cccc r}
\toprule
Method &  R@5 & R@10 & Impr. (avg) \\
\midrule
HLLM-$64$-S$1$  & \textbf{3.39} & \textbf{5.34}& +0.0\%\\
HLLM-$64$-S$8$  & 3.16 & 5.03  & -6.29\% \\
\bottomrule
\end{tabular}
\caption{Performance comparison of HLLM with top-\textit{k} token selection under different selection chunk sizes on the Amazon Books dataset, using Qwen3-Embedding-0.6B as the backbone and truncating text sequences to 64 tokens.}  
\label{app:tbl:topk_logits_selection}
\end{table}

\section{More Related Work}
\subsection{Sequential Recommendation (SR)}
Sequential recommendation (SR) captures the diverse preferences and evolving interests of users by modeling their recent interaction sequences \cite{hidasi2015sessionbased,10.1145/2988450.2988452,10.1609/aaai.v33i01.33015941,kang2018selfattentive,10.1145/3336191.3371786,10.1145/3383313.3412258, hstu, Harte2023LeveragingLL, HLLM}.
\paragraph{Traditional SR} Traditional sequential recommendation methods primarily focus on enhancing neural architectures to more accurately model users' diverse interests. Representative examples include Caser \cite{10.1145/3159652.3159656}, SASRec \cite{kang2018selfattentive}, BERT4Rec \cite{10.1145/3357384.3357895}, CLS2Rec \cite{cls4rec}, and CoSeRec \cite{CoSeRec}. While effective, these approaches often struggle to scale to large parameter counts. To address this limitation, generative recommendation methods such as DreamRec \cite{dreamrec} and HSTU \cite{hstu} have been proposed, which are inherently scalable to larger models. However, these generative approaches typically fail to leverage the rich semantic knowledge pre-encoded in pre-trained LLMs, thereby limiting their potential to further improve recommendation performance. In contrast, our proposed \TheName{} effectively harnesses the full semantic knowledge embedded in pre-trained LLMs, enabling more accurate sequential recommendations.

\subsection{Prompt Engineering and Preference Optimization}
\paragraph{Prompt Engineering}
A substantial body of research shows that fine-grained control over individual tokens can unlock more intelligent, efficient, and trustworthy LLM applications, even under constrained token budgets~\cite{lester2021power, li2021prefix, zhao2025pmpo, deng2022rlprompt}. 
\textit{Prompt Tuning}~\cite{lester2021power} and \textit{Prefix-Tuning}~\cite{li2021prefix} operate in the model's continuous embedding space, using gradient-based optimization to learn effective prompt representations. 
RL-Prompt~\cite{deng2022rlprompt} reformulates discrete prompt optimization as a reinforcement learning problem, where token sequences are selected to maximize task-specific rewards. The method achieves strong performance across a range of NLP tasks and intriguingly reveals that optimal prompts often deviate from standard grammatical conventions to enhance effectiveness.
In contrast to these methods—which optimize prompts to steer model outputs—our work focuses on identifying informative tokens directly from the input.

\paragraph{Preference Optimization}
Preference optimization aims to align large language models (LLMs) with human preferences and values. 
Direct Preference Optimization (DPO) \cite{rafailov2024direct} is a representative method that increases the likelihood of outputs preferred by humans while decreasing the likelihood of disfavored ones. 
DPO has inspired numerous follow-up studies \cite{meng2024simpo, azar24aipo, Lu2024EliminatingBL, Ethayarajh2024KTOMA,hong-etal-2024-orpo,zhu2025acl}. 
Unlike these approaches, our method does not modify the output probability distribution of LLMs. 
Instead, we adjust token-level probabilities in the input space to identify which input segments, when emphasized, can enhance downstream recommendation performance.

\section{Algorithm Framework}

\begin{algorithm}
  \SetKwInOut{Input}{Input}
  \SetKwInOut{Output}{Output}
  
  \Input{Training Instances: $\mathcal{D}_0$, Number of Epochs: $U$, Initial Model: $\mathbf{\pi}_0$. }
  \Output{Final Model: $\mathbf{\pi}_U$, Updated Training Instances:  $\mathcal{D}_U$}
  
  \tcp{Pretrain LLM-based recommendation model $\mathbf{\pi}_1$}
  $\mathbf{\pi}_1 \leftarrow \mathbf{\pi}_0$
  
  \For{$i\leftarrow 1$ \KwTo $U$}{
    $\mathbf{x} \sim \mathcal{D}_0$  \qquad\tcp{Sample training data from $\mathcal{D}_0$}
    $\mathbf{\pi}_1 \leftarrow \textsc{Update} (\mathbf{\pi}_1, \mathbf{x})$
  }
  \tcp{Freeze the LLM backbone and initialize the policy head.}
  $\mathbf{\pi}_\theta \leftarrow \textsc{Freeze-And-Init-Policy} (\mathbf{\pi}_1, \theta)$
  
  \For{$i \leftarrow 1$ \KwTo $U$}{
    $\mathbf{x} \sim \mathcal{D}_0$ \qquad \tcp{Sample training data from $\mathcal{D}_0$}
    $\mathbf{s}_0 \sim \pi_\theta(\mathbf{x})$ \qquad \tcp{Sampling Item Tokens.} 
    $\mathbf{s}_1 \sim \pi_\theta(\mathbf{x})$
    
    $c_0 \leftarrow \textsc{Cross-Entropy}(\mathbf{s}_0, \mathbf{\pi}_1)$
    
    $c_1 \leftarrow \textsc{Cross-Entropy}(\mathbf{s}_1, \mathbf{\pi}_1)$
    
    $r = \textsc{Compute-Reward}(c_0, c_1)$
    
    $\mathbf{\pi}_\theta \leftarrow \textsc{Update}(\mathbf{\pi}_\theta, r)$
  }
  $\mathcal{D}_U \leftarrow \textsc{Select}(\mathbf{\pi}_\theta, \mathcal{D}_0)$
  
  $\mathbf{\pi}_U \leftarrow \mathbf{\pi}_0$

  \For{$i\leftarrow 1$ \KwTo $U$}{
    $\mathbf{x} \sim \mathcal{D}_U$  \qquad\tcp{Sample training data from $\mathcal{D}_U$}
    $\mathbf{\pi}_U \leftarrow \textsc{Update} (\mathbf{\pi}_U, \mathbf{x})$
  }  
  \Return{$\pi_U, \mathcal{D}_U$}
  \LinesNumberedHidden
  \caption{The \TheName{} pipeline}\label{app:algo}
\end{algorithm}
We present the pseudocode of \TheName{} in Algorithm~\ref{app:algo}. 
\TheName{} consists of three main stages. 
In Stage 1, an LLM-based recommendation model (e.g., HLLM) is pretrained on user-item interaction data using textual features. 
In Stage 2, the LLM backbone is frozen, and a policy head is introduced; \TheName{} then learns a token selection policy using the proposed proxy reward. 
In the final stage, the learned policy is applied to select the most informative tokens from the dataset, and the base model is retrained on the refined item representations.

\section{Proof}
\label{app:proof}
\begin{theorem}
\label{prop:core}
Under the framework of \TheName{}, the following properties hold:
\begin{itemize}
  \item Smaller values of the cross-entropy losses $\mathrm{ce}_1$ and $\mathrm{ce}_2$ in Eq.~\eqref{eq:ce-first} and Eq.~\eqref{eq:ce-second} yield a tighter approximation to the ground-truth preference distribution in terms of KL divergence.
  
  \item Token chunks shared between the two sampled sequences $\mathcal{I}'_M$ and $\mathcal{I}''_M$ do not contribute to the gradient updates of the policy parameters $\theta$.
  
  \item The objective in Eq.~\eqref{eq:expected-reward} increases the likelihood of selecting informative token chunks while suppressing less informative ones.
\end{itemize}
\end{theorem}

\begin{proof}
  Let $\mathcal{L} = -\mathcal{R}(\theta)$.
  
(1) We denote the preference distribution associated with $\mathbf{e}'_u$ as $p'$, that associated with $\mathbf{e}''_u$ as $p''$, and the ground-truth preference distribution as $q$. By definition, the Kullback–Leibler (KL) divergences are given by
\begin{align}
  \mathrm{KL}(q \,\|\, p') 
    &= \int q(x) \log \frac{q(x)}{p'(x)} \, \mathrm{d}x \\
    &= \mathrm{ce}_1 - H(q),
\end{align}
and
\begin{align}
  \mathrm{KL}(q \,\|\, p'') 
    &= \int q(x) \log \frac{q(x)}{p''(x)} \, \mathrm{d}x \\
    &= \mathrm{ce}_2 - H(q),
\end{align}
where $H(q) = -\int q(x) \log q(x) \, \mathrm{d}x$ denotes the entropy of $q$. 

Since $H(q)$ is constant with respect to the model parameters, the comparison between $\mathrm{ce}_1$ and $\mathrm{ce}_2$ is equivalent to the comparison between $\mathrm{KL}(q \,\|\, p')$ and $\mathrm{KL}(q \,\|\, p'')$. In other words, a smaller cross-entropy value corresponds to a tighter approximation of the ground-truth preference distribution in terms of KL divergence.  
  
(2) Assume there exists a common chunk $C$ shared between $P_1=P(\mathcal{I}'_M \mid \theta)$ and $P_2=P(\mathcal{I}''_M \mid \theta)$. The gradient of the loss with respect to $C$ is given by:
\begin{align}
\frac{\partial \mathcal{L}}{\partial C} & = -r \frac{\partial (\log P_1 - \log P_2)}{\partial C} \nonumber \\
& = -r \frac{\partial (\log P_1(C) - \log P_2(C))}{\partial C} + 0 \nonumber \\
& = 0
\end{align}
Hence, the gradient $\frac{\partial \mathcal{L}}{\partial C}$ vanishes.

(3) If $\text{ce}_1 < \text{ce}_2$, indicating that the sequence associated with $P_1=P(\mathcal{I}'_M \mid \theta)$ yields lower cross-entropy than $P_2 = P(\mathcal{I}''_M \mid \theta)$ and is thus more informative, we set the reward $r = 1$. In this case, the optimization objective is:
\begin{equation}
  \mathcal{L} = \log P_2 - \log P_1
  \label{eq:ce1}
\end{equation}
Minimizing Eq.~(\ref{eq:ce1}) increases the likelihood of $P_1$ while decreasing that of $P_2$.

Conversely, if $\text{ce}_1 > \text{ce}_2$, implying that $P_2$ corresponds to the more informative sequence, we set $r = -1$. The resulting loss becomes:
\begin{equation}
  \mathcal{L} = \log P_1 - \log P_2
  \label{eq:ce2}
\end{equation}
Minimizing Eq.~(\ref{eq:ce2}) promotes $P_2$ and suppresses $P_1$.

In both scenarios, the training objective effectively increases the probability of the more informative sequence, thereby guiding the policy network to select token chunks that preserve semantically salient information.
\end{proof}

%% file: acl.bib
@article{cls4rec,
title={Contrastive Learning for Sequential Recommendation},
author={Xu Xie and Fei Sun and Zhaoyang Liu and Shiwen Wu and Jinyang Gao and Bolin Ding and Bin Cui},
journal={arXiv preprint arXiv:2010.14395},
year={2021}
}

@article{CoSeRec,
title={Contrastive Self-supervised Sequential Recommendation with Robust Augmentation},
author={Zhiwei Liu and Yongjun Chen and Jia Li and Philip S. Yu and Julian McAuley and Caiming Xiong},
journal={arXiv preprint arXiv:2108.06479},
year={2021}
}

@inproceedings{10.1145/3357384.3357895,
author = {Sun, Fei and Liu, Jun and Wu, Jian and Pei, Changhua and Lin, Xiao and Ou, Wenwu and Jiang, Peng},
title = {BERT4Rec: Sequential Recommendation with Bidirectional Encoder Representations from Transformer},
year = {2019},
booktitle = {Proceedings of the 28th ACM International Conference on Information and Knowledge Management},
pages = {1441–1450},
numpages = {10},
}

@inproceedings{10.1145/3159652.3159656,
author = {Tang, Jiaxi and Wang, Ke},
title = {Personalized Top-N Sequential Recommendation via Convolutional Sequence Embedding},
year = {2018},
booktitle = {Proceedings of the Eleventh ACM International Conference on Web Search and Data Mining},
pages = {565–573},
numpages = {9},
}

@article{dreamrec,
title={Generate What You Prefer: Reshaping Sequential Recommendation via Guided Diffusion},
author={Zhengyi Yang and Jiancan Wu and Zhicai Wang and Xiang Wang and Yancheng Yuan and Xiangnan He},
journal={arXiv preprint arXiv:2310.20453},
year={2023}
}

@inproceedings{hstu,
author = {Zhai, Jiaqi and Liao, Lucy and Liu, Xing and Wang, Yueming and Li, Rui and Cao, Xuan and Gao, Leon and Gong, Zhaojie and Gu, Fangda and He, Jiayuan and Lu, Yinghai and Shi, Yu},
title = {Actions speak louder than words: trillion-parameter sequential transducers for generative recommendations},
year = {2024},
booktitle = {Proceedings of the 41st International Conference on Machine Learning},
}

@article{hidasi2015sessionbased,
  title = {Session-based Recommendations with Recurrent Neural Networks},
  author = {Hidasi, Balazs and Karatzoglou, Alexandros and Baltrunas, Linas and Tikk, Domonkos},
  journal={arXiv preprint arXiv:1511.06939},  
  year = {2015}  
}

@inproceedings{10.1145/2988450.2988452,
author = {Tan, Yong Kiam and Xu, Xinxing and Liu, Yong},
title = {Improved Recurrent Neural Networks for Session-Based Recommendations},
year = {2016},
doi = {10.1145/2988450.2988452},
booktitle = {Proceedings of the 1st Workshop on Deep Learning for Recommender Systems},
pages = {17--22},
numpages = {6},
}

@inproceedings{10.1609/aaai.v33i01.33015941,
author = {Zhou, Guorui and Mou, Na and Fan, Ying and Pi, Qi and Bian, Weijie and Zhou, Chang and Zhu, Xiaoqiang and Gai, Kun},
title = {Deep Interest Evolution Network for Click-through Rate Prediction},
year = {2019},
doi = {10.1609/aaai.v33i01.33015941},
booktitle = {Proceedings of the Thirty-Third AAAI Conference on Artificial Intelligence and Thirty-First Innovative Applications of Artificial Intelligence Conference and Ninth AAAI Symposium on Educational Advances in Artificial Intelligence},
}

@inproceedings{rendle10mc,
author = {Rendle, Steffen and Freudenthaler, Christoph and Schmidt-Thieme, Lars},
title = {Factorizing Personalized Markov Chains for Next-Basket Recommendation},
year = {2010},
booktitle = {WWW},
pages = {811--820},
numpages = {10}
}

@inproceedings{cheng13poi,
author = {Cheng, Chen and Yang, Haiqin and Lyu, Michael R. and King, Irwin},
title = {Where You like to Go next: Successive Point-of-Interest Recommendation},
year = {2013},
isbn = {9781577356332},
booktitle = {IJCAI},
pages = {2605--2611},
numpages = {7}
}

@inproceedings{kang2018selfattentive,
  author = {Kang, Wang-Cheng and McAuley, Julian J.},
  booktitle = {ICDM},
  pages = {197--206},
  title = {Self-Attentive Sequential Recommendation},
  year = {2018}
}

@inproceedings{10.1145/3336191.3371786,
author = {Li, Jiacheng and Wang, Yujie and McAuley, Julian},
title = {Time Interval Aware Self-Attention for Sequential Recommendation},
year = {2020},
booktitle = {WSDM},
pages = {322--330},
numpages = {9}
}

@inproceedings{10.1145/3383313.3412258,
author = {Wu, Liwei and Li, Shuqing and Hsieh, Cho-Jui and Sharpnack, James},
title = {SSE-PT: Sequential Recommendation Via Personalized Transformer},
year = {2020},
booktitle = {RecSys},
pages = {328--337},
numpages = {10}
}

@article{Harte2023LeveragingLL,
  title={Leveraging Large Language Models for Sequential Recommendation},
  author={Jesse Harte and Wouter Zorgdrager and Panos Louridas and Asterios Katsifodimos and D. Jannach and Marios Fragkoulis},
  journal={Proceedings of the 17th ACM Conference on Recommender Systems},
  year={2023},
}

@article{HLLM,
title={HLLM: Enhancing Sequential Recommendations via Hierarchical Large Language Models for Item and User Modeling},
author={Junyi Chen and Lu Chi and Bingyue Peng and Zehuan Yuan},
journal={arXiv preprint arXiv:2409.12740},
year={2024}
}

@article{HLLM-Creator,
title={HLLM-Creator: Hierarchical LLM-based Personalized Creative Generation},
author={Junyi Chen and Lu Chi and Siliang Xu and Shiwei Ran and Bingyue Peng and Zehuan Yuan},
journal={arXiv preprint arXiv:2508.18118},
year={2025}
}

@article{liu2024large,
  title={Large Language Model Empowered Embedding Generator for Sequential Recommendation},
  author={Liu, Qidong and Wu, Xian and Wang, Wanyu and Wang, Yejing and Zhu, Yuanshao and Zhao, Xiangyu and Tian, Feng and Zheng, Yefeng},
  journal={arXiv preprint arXiv:2409.19925},
  year={2024}
}

@article{alphafuse,
  title={AlphaFuse: Learn ID Embeddings for Sequential Recommendation in Null Space of Language Embeddings},
  author={Guoqing Hu and An Zhang and Shuo Liu and Zhibo Cai and Xun Yang and Xiang Wang},
  journal={arXiv preprint arXiv:2504.19218},
  year={2025}
}

@article{Zhang2024TinyLlamaAO,
  title={TinyLlama: An Open-Source Small Language Model},
  author={Peiyuan Zhang and Guangtao Zeng and Tianduo Wang and Wei Lu},
  journal={arXiv preprint arXiv: 2401.02385},
  year={2024},
}

@article{baichuan2,
  title={Baichuan 2: Open Large-scale Language Models},
  author={Baichuan},
  journal={arXiv preprint arXiv: 2309.10305},
  year={2023},
}

@article{qwen3embedding,
  title={Qwen3 Embedding: Advancing Text Embedding and Reranking Through Foundation Models},
  author={Zhang, Yanzhao and Li, Mingxin and Long, Dingkun and Zhang, Xin and Lin, Huan and Yang, Baosong and Xie, Pengjun and Yang, An and Liu, Dayiheng and Lin, Junyang and Huang, Fei and Zhou, Jingren},
  journal={arXiv preprint arXiv:2506.05176},
  year={2025}
}

@inproceedings{10.1145/2766462.2767755,
author = {McAuley, Julian and Targett, Christopher and Shi, Qinfeng and van den Hengel, Anton},
title = {Image-Based Recommendations on Styles and Substitutes},
year = {2015},
doi = {10.1145/2766462.2767755},
booktitle = {Proceedings of the 38th International ACM SIGIR Conference on Research and Development in Information Retrieval},
pages = {43–52},
numpages = {10},
}

@article{cheng2023image,
  title={An Image Dataset for Benchmarking Recommender Systems with Raw Pixels},
  author={Cheng, Yu and Pan, Yunzhu and Zhang, Jiaqi and Ni, Yongxin and Sun, Aixin and Yuan, Fajie},
  journal={arXiv preprint arXiv:2309.06789},
  year={2023}
}

@article{infonce,
  title={Representation Learning with Contrastive Predictive Coding},
  author={Aaron van den Oord and Yazhe Li and Oriol Vinyals},
  journal={arXiv preprint arXiv:1807.03748},  
  year={2018}
}

@article{lester2021power,
  title={The power of scale for parameter-efficient prompt tuning},
  author={Lester, Brian and Al-Rfou, Rami and Constant, Noah},
  journal={arXiv preprint arXiv:2104.08691},
  year={2021}
}

@article{li2021prefix,
  title={Prefix-tuning: Optimizing continuous prompts for generation},
  author={Li, Xiang Lisa and Liang, Percy},
  journal={arXiv preprint arXiv:2101.00190},
  year={2021}
}

@article{zhao2025pmpo,
  title={PMPO: Probabilistic Metric Prompt Optimization for Small and Large Language Models},
  author={Zhao, Chenzhuo and Liu, Ziqian and Wang, Xingda and Lu, Junting and Ruan, Chaoyi},
  journal={arXiv preprint arXiv:2505.16307},
  year={2025}
}

@article{deng2022rlprompt,
  title={Rlprompt: Optimizing discrete text prompts with reinforcement learning},
  author={Deng, Mingkai and Wang, Jianyu and Hsieh, Cheng-Ping and Wang, Yihan and Guo, Han and Shu, Tianmin and Song, Meng and Xing, Eric P and Hu, Zhiting},
  journal={arXiv preprint arXiv:2205.12548},
  year={2022}
}

@inproceedings{sun2024large,
  title={Large language models enhanced collaborative filtering},
  author={Sun, Zhongxiang and Si, Zihua and Zang, Xiaoxue and Zheng, Kai and Song, Yang and Zhang, Xiao and Xu, Jun},
  booktitle={Proceedings of the 33rd ACM International Conference on Information and Knowledge Management},
  pages={2178--2188},
  year={2024}
}

@inproceedings{zhang2024notellm,
  title={Notellm: A retrievable large language model for note recommendation},
  author={Zhang, Chao and Wu, Shiwei and Zhang, Haoxin and Xu, Tong and Gao, Yan and Hu, Yao and Chen, Enhong},
  booktitle={Companion Proceedings of the ACM Web Conference 2024},
  pages={170--179},
  year={2024}
}

@inproceedings{zhang2025notellm,
  title={Notellm-2: Multimodal large representation models for recommendation},
  author={Zhang, Chao and Zhang, Haoxin and Wu, Shiwei and Wu, Di and Xu, Tong and Zhao, Xiangyu and Gao, Yan and Hu, Yao and Chen, Enhong},
  booktitle={Proceedings of the 31st ACM SIGKDD Conference on Knowledge Discovery and Data Mining},
  pages={2815--2826},
  year={2025}
}

@article{qiao2024llm4sbr,
  title={LLM4SBR: A lightweight and effective framework for integrating large language models in session-based recommendation},
  author={Qiao, Shutong and Gao, Chen and Wen, Junhao and Zhou, Wei and Luo, Qun and Chen, Peixuan and Li, Yong},
  journal={arXiv preprint arXiv:2402.13840},
  year={2024}
}

@article{liu2024llm,
  title={Llm-esr: Large language models enhancement for long-tailed sequential recommendation},
  author={Liu, Qidong and Wu, Xian and Wang, Yejing and Zhang, Zijian and Tian, Feng and Zheng, Yefeng and Zhao, Xiangyu},
  journal={Advances in Neural Information Processing Systems},
  volume={37},
  pages={26701--26727},
  year={2024}
}

@article{xu2024slmrec,
  title={SLMRec: Distilling large language models into small for sequential recommendation},
  author={Xu, Wujiang and Wu, Qitian and Liang, Zujie and Han, Jiaojiao and Ning, Xuying and Shi, Yunxiao and Lin, Wenfang and Zhang, Yongfeng},
  journal={arXiv preprint arXiv:2405.17890},
  year={2024}
}

@inproceedings{han2024lm,
  title={LM-Infinite: Zero-Shot Extreme Length Generalization for Large Language Models},
  author={Han, Chi and Wang, Qifan and Peng, Hao and Xiong, Wenhan and Chen, Yu and Ji, Heng and Wang, Sinong},
  booktitle={Proceedings of the 2024 Conference of the North American Chapter of the Association for Computational Linguistics: Human Language Technologies (Volume 1: Long Papers)},
  pages={3991--4008},
  year={2024}
}

@inproceedings{rafailov2024direct,
  title={Direct preference optimization: Your language model is secretly a reward model},
  author={Rafailov, Rafael and Sharma, Archit and Mitchell, Eric and Manning, Christopher D and Ermon, Stefano and Finn, Chelsea},
  booktitle={NeurIPS},
  year={2023}
}

@inproceedings{meng2024simpo,
   title={SimPO: Simple Preference Optimization with a Reference-Free Reward},
   author={Meng, Yu and Xia, Mengzhou and Chen, Danqi},
   booktitle={NeurIPS},
   year={2024}
}

@InProceedings{azar24aipo,
  title={A General Theoretical Paradigm to Understand Learning from Human Preferences},
  author={Gheshlaghi Azar, Mohammad and Daniel Guo, Zhaohan and Piot, Bilal and Munos, Remi and Rowland, Mark and Valko, Michal and Calandriello, Daniele},
  booktitle={AISTATS},
  year={2024},
}

@article{Lu2024EliminatingBL,
  title={Eliminating Biased Length Reliance of Direct Preference Optimization via Down-Sampled KL Divergence},
  author={Junru Lu and Jiazheng Li and Siyu An and Meng Zhao and Yulan He and Di Yin and Xing Sun},
  journal={arXiv preprint arXiv:2406.10957},  
  year={2024},
}

@article{Ethayarajh2024KTOMA,
  title={KTO: Model Alignment as Prospect Theoretic Optimization},
  author={Kawin Ethayarajh and Winnie Xu and Niklas Muennighoff and Dan Jurafsky and Douwe Kiela},
  year={2024},
  journal={arXiv preprint arXiv:2402.01306},  
}

@inproceedings{hong-etal-2024-orpo,
    title = "{ORPO}: Monolithic Preference Optimization without Reference Model",
    author = {Hong, Jiwoo  and
      Lee, Noah  and
      Thorne, James},
    booktitle = {EMNLP},
    year = "2024",
}

@inproceedings{10.1145/3637528.3671931,
author = {Kim, Sein and Kang, Hongseok and Choi, Seungyoon and Kim, Donghyun and Yang, Minchul and Park, Chanyoung},
title = {Large Language Models meet Collaborative Filtering: An Efficient All-round LLM-based Recommender System},
year = {2024},
isbn = {9798400704901},
publisher = {Association for Computing Machinery},
address = {New York, NY, USA},
url = {https://doi.org/10.1145/3637528.3671931},
doi = {10.1145/3637528.3671931},
pages = {1395–1406},
numpages = {12},
location = {Barcelona, Spain},
series = {KDD '24},
booktitle = {Proceedings of the 30th ACM SIGKDD Conference on Knowledge Discovery and Data Mining}
}

@inproceedings{10.1145/3626772.3657690,
author = {Liao, Jiayi and Li, Sihang and Yang, Zhengyi and Wu, Jiancan and Yuan, Yancheng and Wang, Xiang and He, Xiangnan},
title = {LLaRA: Large Language-Recommendation Assistant},
year = {2024},
isbn = {9798400704314},
publisher = {Association for Computing Machinery},
address = {New York, NY, USA},
url = {https://doi.org/10.1145/3626772.3657690},
doi = {10.1145/3626772.3657690},
pages = {1785–1795},
numpages = {11},
location = {Washington DC, USA},
series = {SIGIR '24},
booktitle = {Proceedings of the 47th International ACM SIGIR Conference on Research and Development in Information Retrieval}
}

@inproceedings{6c48a0f1d7d84077a16ce55105dc8ddc,
title = "Harnessing Large Language Models for Text-Rich Sequential Recommendation",
author = "Zhi Zheng and Chao, \{Wen Shuo\} and Zhaopeng Qiu and Hengshu Zhu and Hui Xiong",
year = "2024",
month = may,
day = "13",
doi = "10.1145/3589334.3645358",
series = "WWW 2024 - Proceedings of the ACM Web Conference",
publisher = "Association for Computing Machinery, Inc",
pages = "3207--3216",
booktitle = "WWW 2024 - Proceedings of the ACM Web Conference",
}

@inproceedings{10.1145/3746252.3761507,
author = {Zhu, Jie and Fan, Zhifang and Zhu, Xiaoxie and Jiang, Yuchen and Wang, Hangyu and Han, Xintian and Ding, Haoran and Wang, Xinmin and Zhao, Wenlin and Gong, Zhen and Yang, Huizhi and Chai, Zheng and Chen, Zhe and Zheng, Yuchao and Chen, Qiwei and Zhang, Feng and Zhou, Xun and Xu, Peng and Yang, Xiao and Wu, Di and Liu, Zuotao},
title = {RankMixer: Scaling Up Ranking Models in Industrial Recommenders},
year = {2025},
isbn = {9798400720406},
publisher = {Association for Computing Machinery},
address = {New York, NY, USA},
url = {https://doi.org/10.1145/3746252.3761507},
doi = {10.1145/3746252.3761507},
numpages = {8},
location = {Seoul, Republic of Korea},
series = {CIKM '25},
booktitle = {Proceedings of the 34th ACM International Conference on Information and Knowledge Management}
}

@misc{zhu2025csdm,
      title={Addressing Cold-start Problem in Click-Through Rate Prediction via Supervised Diffusion Modeling}, 
      author={Wenqiao Zhu and Lulu Wang and Jun Wu},
      year={2025}
}

@misc{zhu2025acl,
     title={SGDPO: Self-Guided Direct Preference Optimization for Language Model Alignment},
     author = {Wenqiao Zhu and Ji Liu and Lulu Wang and Jun Wu and Yulun Zhang},
     year={2025}
}

@inproceedings{10.1145/3511808.3557704,
author = {Zhu, Wenqiao and Xu, Yesheng and Huang, Xin and Min, Qiyang and Zhou, Xun},
title = {Spherical Graph Embedding for Item Retrieval in Recommendation System},
year = {2022},
pages = {4752–4756},
series = {CIKM22},
booktitle = {Proceedings of the 31st ACM International Conference on Information and Knowledge Management}
}
